\documentclass[aos,preprint]{imsart}
\usepackage[]{graphicx}\usepackage[]{color}

\makeatletter
\def\maxwidth{ %
  \ifdim\Gin@nat@width>\linewidth
    \linewidth
  \else
    \Gin@nat@width
  \fi
}
\makeatother

\definecolor{fgcolor}{rgb}{0.345, 0.345, 0.345}

\usepackage{framed}
\makeatletter
 {\par\unskip\endMakeFramed%
 \at@end@of@kframe}
\makeatother

\definecolor{shadecolor}{rgb}{.97, .97, .97}
\definecolor{messagecolor}{rgb}{0, 0, 0}
\definecolor{warningcolor}{rgb}{1, 0, 1}
\definecolor{errorcolor}{rgb}{1, 0, 0}
\newenvironment{knitrout}{}{} 

\usepackage{alltt} 
\usepackage[latin1]{inputenc}
\usepackage{enumitem}
\usepackage{todonotes}
\usepackage{latexsym}
\usepackage{amssymb}
\usepackage{epsfig}
\usepackage{amsmath}
\usepackage{xcolor}
\usepackage{float}
\usepackage{hyperref}

\IfFileExists{upquote.sty}{\usepackage{upquote}}{}

\usepackage[latin1]{inputenc}
\usepackage{enumitem}
\usepackage{todonotes}
\RequirePackage[OT1]{fontenc}
\RequirePackage{amsmath,amssymb,amsthm}
\RequirePackage[numbers]{natbib}
\RequirePackage[colorlinks,citecolor=blue,urlcolor=blue]{hyperref}
\usepackage{graphicx}

\startlocaldefs
\numberwithin{equation}{section}
\theoremstyle{plain}

\endlocaldefs

\theoremstyle{plain}
\long\def\comment#1{}

\newtheorem{theorem}{Theorem}

\newtheorem{lemma}{Lemma}

\theoremstyle{definition}

\numberwithin{definition}{section}

\newtheorem{remark}{Comment}[section]
\numberwithin{remark}{section}

\DeclareMathOperator{\Var}{Var}

\newcommand{\F}{\mathcal{F}}

\renewcommand{\qed}{\hfill{\tiny \ensuremath{\blacksquare} }}%

\newcommand{\E}{\mathbb{E}}

\begin{document}

\begin{frontmatter}
\title{Transformation Models in High-Dimensions\thanksref{T1}}
\runtitle{Transformation Models in High-Dimensions}
\thankstext{T1}{Version September 2017. We thank Philipp Bach for processing the dataset for the application.}

\begin{aug}
\author{\fnms{Sven} \snm{Klaassen}\ead[label=e1]{}}
\and
\author{\fnms{Jannis} \snm{K\"uck}\ead[label=e2]{}}
\and
\author{\fnms{Martin} \snm{Spindler}\ead[label=e3]{}}

\affiliation{University of Hamburg}

\address{Jannis K\"uck\\
University of Hamburg\\
Hamburg Business School\\
Moorweidenstr. 18\\
20148 Hamburg\\
Germany\\
E-mail: jannis.kueck@uni-hamburg.de}

\address{Sven Klaassen\\
University of Hamburg\\
Hamburg Business School\\
Moorweidenstr. 18\\
20148 Hamburg\\
Germany\\
E-mail: sven.klaassen@uni-hamburg.de}

\address{Martin Spindler\\
University of Hamburg\\
Hamburg Business School\\
Moorweidenstr. 18\\
20148 Hamburg\\
Germany\\
E-mail: martin.spindler@uni-hamburg.de}
\end{aug}

\begin{abstract}
Transformation models are a very important tool for applied statisticians and econometricians. In many applications, the dependent variable is transformed so that homogeneity or normal distribution of the error holds. In this paper, we analyze transformation models in a high-dimensional setting, where the set of potential covariates is large. We propose an estimator for the transformation parameter and we show that it is asymptotically normally distributed using an orthogonalized moment condition where the nuisance functions depend on the target parameter. In a simulation study, we show that the proposed estimator works well in small samples. A common practice in labor economics is to transform wage with the log-function. In this study, we test if this transformation holds in CPS data from the United States.

\end{abstract}



\end{frontmatter}

\section{Introduction}
Over the last few years, substantial progress has been made in the problem of fitting high-dimensional linear models of the form 
\begin{align}\label{mo}
Y=X\beta +\varepsilon,
\end{align}
where the number of regressors $p$ is much larger than the sample size $n$. The theoretical properties of penalization approaches, such as lasso, are now well understood under the assumption that the coefficient vector $\beta$ is sparse. A detailed summary of the recent results is given in textbook length in \cite{nr:buhlmann.de-geer:2011}.\\
In this paper, we take up the idea of the high dimensional linear model in \ref{mo} and combine it with a parametric transformation of the response variable $\Lambda_\theta(\cdot)\in\mathcal{F}_\Lambda$, where $\mathcal{F}_\Lambda=\{\Lambda_\theta(\cdot):\theta\in\Theta\}$ is a given family of strictly monotone increasing functions. For every $\theta\in\Theta$, we assume a linear model
\begin{align}\label{trmo}
\Lambda_{\theta}(Y)=X\beta_{\theta}+\varepsilon_\theta
\end{align}
with $\E[\varepsilon_\theta]=0$. Our analysis allows the number of regressors to be much larger than the number of observations, although we require sparsity for every $\beta_\theta$ in \ref{trmo}. The goal of data transformation is to change the scale preventing wrong model assumptions, such as by establishing normal distributed errors.\\
Transformation models are widely used in biostatistics, medicine, economics and many other fields, where skewed data is often first log-transformed and then analyzed. For example, in labor economics it is common practice to transform the wage by the logarithm. Feng et al. \cite{feng2013log} list some common scenarios of the misuse and misinterpretation of log transformation. This underlines the importance of the right transformation to handle the problem of skewed data and non-negative outcomes.\\ 
In this study, we will present an estimate for the unknown transformation parameter $\theta_0\in\Theta$ in a high-dimensional transformation model, which satisfies
\begin{align}
\Lambda_{\theta_0}(Y)=X\beta_{\theta_0}+\varepsilon_{\theta_0}
\end{align}
with $\varepsilon_{\theta_0}\sim\mathcal{N}(0,\sigma^2)$. This means that under the true parameter $\theta_0$, the errors are normally distributed with unknown variance.
We establish that our estimator is root-n consistent, asymptotically unbiased, and normal. The transformation enables us to establish normality of the error terms and subsequent application of procedures based on normality.
Our setting fits into a general Z-estimation problem with a high-dimensional nuisance function, which depends on the target parameter $\theta$. Inference on a target parameter in general Z-estimation problems in high dimensions is covered in Belloni et al. (2014) \cite{belloni2014uniform} and Chernozhukov et al. (2017) \cite{chernozhukov2017double}. In high-dimensional transformation models, the nuisance function depends on the target parameter $\theta$; therefore, in the supplementary material we establish a theorem regarding inference in a general Z-estimation setting under a different set of entropy conditions where such a dependence is explicitly allowed. This result might be of interest for Z-estimation problems with the same structure. Appendices \ref{appendixB} and \ref{appendixC} might be of independent interest for related problems of the same underlying structure.

Transformation of the dependent variable is very common in statistics and econometrics. The Box-Cox power transformations (\cite{BoxCox}) or the modification proposed by Yeo and Johnson (\cite{YeoJohnson}) are very popular transformations. The aim of transformations is typically to achieve symmetry, normality, or independence of the error terms. In labor economics the analysis of wage data is key, and wage data is non-negative and often highly skewed. By default, wage data are transformed by the logarithm and then further processed; for example, as a dependent variable in a Mincer equation. A crucial point for the subsequent analysis is that the applied transformation is correctly specified. In this paper, we analyse the estimation and inference of the transformation parameter in a high-dimensional setting; that is, where the set of potential dependent variables is large, even larger than the number of observations. Here we focus on estimation and inference of the transformation parameter because this is the first crucial step and it is important for the interpretation of the model and application of subsequent statistical procedures. A related line of research has focused on inference on the covariates in the model. Given that inference in this case relies on the estimated transformation model, valid post-selection/estimation inference is crucial, as pointed out by \cite{BD1981}, which has led to a vivid discussion on this topic. \cite{BD1981} cover the parametric case, the semiparametric case is covered by \cite{linton2008estimation}, and has more recently been examined by \cite{KL2017}, amongst others. Inference on the covariates in high-dimensional settings is an interesting problem that we plan to address in future research. The underlying theory is built on Neyman orthogonal moment conditions, as summarized in \cite{chernozhukov2017double}.
\section*{Plan of this Paper}
The rest of this paper is organized as follows. In Section \ref{normal}, we formally define the setting and propose an estimator for the transformation parameter. In section \ref{main section}, we prove that a Neyman orthogonality condition holds and we provide theoretical results for the estimation rates of the nuisance functions. We also present the main result for the asymptotic distribution of the estimated transformation parameter. Section \ref{simulation} provides a simulation study and section \ref{application} gives an empirical application.
The proofs are provided in Appendix \ref{appendixA}. The supplementary material includes additional technical material. In Appendix \ref{appendixB}, conditions for the uniform convergence rates of the lasso estimator are presented. Finally, Appendix \ref{appendixC} provides a theoretical result about inference on a target parameter in general Z-problems with dependent and high-dimensional nuisance functions.
\section*{Notation}
In what follows, we work with triangular array data $\{(z_{i,n},i=1,...,n),n=1,2,3,...\}$ with $z_{i,n}=(y_{i,n},x_{i,n})$ defined on some common probability space $(\Omega,\mathcal{A},P)$. The law $P_n\in\mathcal{P}_n$ of $\{z_i,i=1,...,n\}$ changes with $n$. Thus, all parameters that characterize the distribution of $\{z_i,i=1,...,n\}$ are implicitly indexed by the sample size $n$, but we omit the index $n$ to simplify notation.\\
The $l_2$ and $l_1$ norms are denoted by $||\cdot||_2$ and $||\cdot||_1$. The $l_0$-norm, $||\cdot||_0$, denotes the number of non-zero components of a vector. We use the notation $a\vee b:=\max (a,b)$ and $a\wedge b:=\min (a,b)$.\\
The symbol $\E$ denotes the expectation operator with respect to a generic probability measure. If we need to signify the dependence on a probability measure $P$, then we use $P$ as a subscript in $\E_P$. For random variables $Z_1,\dots,Z_n$ and a function $g:\mathcal{Z}\to\mathbb{R}$, we define the empirical expectation
$$\E_n[g(Z)]\equiv\E_{\mathbb{P}_n}[g(Z)]:=\frac{1}{n}\sum\limits_{i=1}^ng(Z_i)$$
and
$$G_n(g):=\frac{1}{\sqrt{n}}\sum\limits_{i=1}^n \Big(g(Z_i)-\E \big[g(Z_i)\big]\Big).$$\\
For a class of measurable functions $\mathcal{F}$ on a measurable space, let
$N(\varepsilon,\mathcal{F},\|\cdot\|)$ be the minimal number of balls $B_{\varepsilon}(g):=\{f:\|g-f\|<\varepsilon\}$ of radius $\varepsilon$ to cover the set $\F$. Let $F$ be an envelope function of $\F$ with $F(x) \ge |f(x)|$ for all $f \in \F$. The uniform entropy number with respect to the $L_r(Q)$ seminorm $||\cdot||_{Q,r}$ is defined as 
$$ent(\mathcal{F},\varepsilon):=\sup_{Q}\log N(\varepsilon\|F\|_{Q,r},\mathcal{F},L_r(Q)),$$
where the supremum is taken over all probability measures $Q$ with $0<\E_Q[F^r]^{1/r}<\infty$.\\
For any function $\nu(\theta,u)$ we use the notation $\dot\nu_{\theta^*}(u):=\partial\nu(\theta,u)/\partial\theta|_{\theta=\theta^*}$, respectively $\nu'_{\theta}(u^*):=\partial\nu(\theta,u)/\partial u|_{u=u^*}$.
\section{Transformation model}\label{normal}
We consider a high-dimensional transformation model, where the unknown transformation parameter $\theta_0$ is identified as being the only parameter for which the errors are normally distributed. This assumption is typical for transformation models. Let $\{\Lambda_\theta(\cdot):\theta\in\Theta\}$ be a given parametric family of strict monotone increasing and two times differentiable functions and $\Theta\subset\mathbb{R}$ be compact. For every $\theta\in\Theta$, we assume a linear model
\begin{align}\label{model}
\Lambda_{\theta}(Y)=X\beta_{\theta}+\varepsilon_\theta
\end{align}
with $\E[\varepsilon_\theta]=0$. We write
\begin{align*}
\varepsilon_\theta:=\Lambda_\theta(Y)-\underbrace{m_\theta(x)}_{=X\beta_{\theta}}
\end{align*}
with
\begin{align*}
m_\theta(x)\equiv m(\theta,x):=E[\Lambda_\theta(Y)|X=x].
\end{align*}
By assumption, 
\begin{align*}
\sigma^2_\theta\equiv\sigma^2(\theta):=Var(\varepsilon_\theta|X)
\end{align*}
is independent from $X$. We allow the number of covariates $p=p_n$ to increase with the sample size $n$ but we require that the index set $$S_\theta:=\{j:\beta_{\theta,j}\neq 0\}$$ is sparse for every $\theta\in\Theta$ with $s:=\sup_{\theta\in\Theta}||\beta_\theta||_0$. The number of relevant variables $s=s_n$ may also increase with the sample size $n$ but it does so at a moderate rate.
We assume that $\beta_\theta$ is differentiable in $\theta$. Therefore, we can write
\begin{align}\label{modeldot}
\dot\Lambda_{\theta}(Y)=X\dot\beta_{\theta}+\dot\varepsilon_\theta
\end{align}
with $\E[\dot\varepsilon_\theta]=0$ under regularity conditions (as mentioned later on). The model \ref{modeldot} is sparse with $\dot s:=\sup_{\theta\in\Theta}||\dot\beta_\theta||_0\le 2s$.\\ \\
The assumption that $\beta_\theta$ is sparse and differentiable is common in other applications. For example, in high-dimensional quantile regression, Belloni and Chernozhukov (2011) \cite{belloni2011l1} assume that for every quantile $u\in(0,1)$ the coefficient $\beta(u)$ is sparse and smooth with respect to $u$.\\ \\
We estimate $\theta_0$ by a method similar to the a "profile likelihood procedure" proposed in Linton et al. (2008) \cite{linton2008estimation}. The main idea is to formulate our estimation problem into a Z-estimation problem and then plug-in estimates for all unknown terms.
\subsection{Transformation parameter}\label{transformation parameter}
For the estimation of the transformation parameter, we first determine the likelihood. Because $\Lambda_{\theta}(\cdot)$ is strictly increasing, we have
\begin{align*}
P(Y\le y|X)&=P(\Lambda_{\theta}(Y)\le \Lambda_{\theta}(y)|X)\\&=P(\varepsilon_\theta\le \Lambda_{\theta}(y)-m_{\theta}(X)|X).
\end{align*}
For $\theta=\theta_0$ we obtain
\begin{align*}
P(Y\le y|X)&=P(\varepsilon_{\theta_0}\le \Lambda_{\theta_0}(y)-m_{\theta_0}(X)|X)\\
&=P(\varepsilon_{\theta_0}\le \Lambda_{\theta_0}(y)-m_{\theta_0}(X))\\
&=\Phi\left(\frac{\Lambda_{\theta_0}(y)-m_{\theta_0}(X)}{\sigma}\right),
\end{align*}
with $\Phi$ being the cdf of a standard normal distribution and $\sigma\equiv\sigma_{\theta_0}$.\\
By transforming the densities, we obtain
\begin{align*}
f_{Y|X}(y|x)&=f_{\varepsilon_{\theta_0}}(\Lambda_{\theta_0}(y)-m_{\theta_0}(x))\Lambda_{\theta_0}'(y)\\
&=\frac{1}{\sqrt{2\pi\sigma^2}}\exp\left(-\frac{(\Lambda_{\theta_0}(y)-m_{\theta_0}(x))^2}{2\sigma^2}\right)\Lambda_{\theta_0}'(y)
\end{align*}
and, therefore, the following log-likelihood function
\begin{align*}
l_{Y|X}(\theta)&=-\frac{n}{2}\log(2\pi\sigma_\theta^2)-\frac{1}{2\sigma_\theta^2}\sum_{i=1}^n(\Lambda_{\theta}(Y_i)-m_{\theta}(X_i))^2+\sum_{i=1}^n\log(\Lambda_\theta'(Y_i)).
\end{align*}
The maximum likelihood estimator
\begin{align}\label{profile}
\begin{aligned}
\theta^*&=\underset{\theta\in\Theta}{\arg \max}\bigg[-\frac{1}{2}\log(2\pi\sigma_\theta^2)-\frac{1}{2\sigma_\theta^2n}\sum_{i=1}^n(\Lambda_{\theta}(Y_i)-m_{\theta}(X_i))^2\\&\quad+\frac{1}{n}\sum_{i=1}^n\log(\Lambda_\theta'(Y_i))\bigg]
\end{aligned}
\end{align}
fulfills
\begin{align*}
0&=\partial\bigg(-\frac{1}{2}\log(2\pi\sigma_\theta^2)-\frac{1}{2\sigma_\theta^2 n}\sum_{i=1}^n(\Lambda_{\theta}(Y_i)-m_{\theta}(X_i))^2\\&\quad+\frac{1}{n}\sum_{i=1}^n\log(\Lambda_\theta'(Y_i))\bigg)/\partial\theta\bigg|_{\theta=\theta^*}\\
&=\frac{1}{n}\sum_{i=1}^n\bigg[-\frac{{\dot\sigma_{\theta^*}}^2}{2\sigma_{\theta^*}^2}-\frac{1}{\sigma_{\theta^*}^2}(\Lambda_{\theta^*}(Y_i)-m_{\theta^*}(X_i))(\dot\Lambda_{\theta^*}(Y_i)-\dot m_{\theta^*}(X_i))\\&\quad+\frac{{\dot\sigma_{\theta^*}}^2}{2\sigma_{\theta^*}^4}(\Lambda_{\theta^*}(Y_i)-m_{\theta^*}(X_i))^2+\frac{\dot\Lambda'_{\theta^*}(Y_i)}{\Lambda'_{\theta^*}(Y_i)}\bigg]\\
&=:\mathbb{E}_n\Big[\psi\big((Y,X),\theta^*,h_0(\theta^*,X)\big)\Big],
\end{align*}
where $h_0:\Theta\times\mathcal{X}\to\mathbb{R}\times\mathbb{R}^+\times \mathbb{R}\times \mathbb{R}$ with
$$h_0\equiv(h_{0,1},h_{0,2},h_{0,3},h_{0,4}):=(m_\theta,\sigma^2_\theta,\dot m_\theta,\dot\sigma^2_\theta)$$
is a nuisance function.
We substitute the function $h_0$ by a Lasso estimator $\hat{h}_0$, which is defined in subsection \ref{nuissance} and analyzed in subsection \ref{uniformsection}.\\
Finally, we estimate the transformation parameter $\theta_0$ by an estimator $\hat{\theta}$, which solves
\begin{align}\label{estimator}
\left|\mathbb{E}_n\Big[\psi\big((Y,X),\hat{\theta},\hat{h}_0(\hat{\theta},X)\big)\Big]\right|=\inf_{\theta\in\Theta}\left|\mathbb{E}_n\Big[\psi\big((Y,X),\theta,\hat{h}_0(\theta,X)\big)\Big]\right|+\epsilon_n,
\end{align}
where $\epsilon_n=o\left(n^{-1/2}\right)$ is the numerical tolerance. 


\subsection{Nuisance function}\label{nuissance}
The unknown nuisance function 
\begin{align*}
h_0=(m_\theta,\sigma^2_\theta,\dot m_\theta,\dot\sigma^2_\theta)
\end{align*}
can be estimated by
\begin{align*}
\hat{h}_0=(\hat{m}_\theta,\hat{\sigma}^2_\theta,\hat{\dot m}_\theta,\hat{\dot\sigma}^2_\theta),
\end{align*}
where $\hat{m}_\theta(x)=x\hat{\beta}_\theta$ with $\hat{\beta}_\theta$ being the lasso estimate
\begin{align*}
\arg\max_{\beta}\mathbb{E}_n[(\Lambda_\theta(Y)-x\beta)^2]+\frac{\lambda}{n}||\Psi_\theta\beta||_1
\end{align*}
with penalty term $\lambda$ and penalty loadings $\Psi_\theta$ as in \cite{belloni2017program} (p. 260). Analogously, we estimate $\dot m_\theta$ by $\hat{\dot m}_\theta(x)=x\hat{\dot\beta}_\theta$ with $\hat{\dot\beta}_\theta$ being the lasso estimate
\begin{align*}
\arg\max_{\beta}\mathbb{E}_n[(\dot\Lambda_\theta(Y)-x\beta)^2]+\frac{\tilde{\lambda}}{n}||\tilde{\Psi}_\theta\beta||_1.
\end{align*}
The unknown variance $\sigma^2_\theta$ can be estimated by 
$$\hat{\sigma}_\theta^2:=\frac{1}{n}\sum\limits_{i=1}^n\hat{\varepsilon}_{i,\theta}^2$$ and $\dot\sigma^2_\theta$ by
$$\hat{\dot\sigma}_\theta^2:=\frac{2}{n}\sum\limits_{i=1}^n\hat{\varepsilon}_{i,\theta}\hat{\dot\varepsilon}_{i,\theta}$$
with $\hat{\varepsilon}_{i,\theta}:=\Lambda_{\theta}(Y_i)-\hat{m}_\theta(X_i)$ and $\hat{\dot\varepsilon}_{i,\theta}:=\dot\Lambda_{\theta}(Y_i)-\hat{\dot m}_\theta(X_i)$. Under regularity conditions
$$\dot{\sigma}_\theta^2=\partial\E[\varepsilon_\theta^2]/\partial\theta=\E[\partial(\varepsilon_\theta^2)/\partial\theta]=2\E[\varepsilon_\theta\dot\varepsilon_\theta]$$
holds.

\subsection{Identification of the true transformation}\label{identificationsec}
First, we formulate our estimation problem as a Z-estimation problem (cf \ref{transformation parameter}). Let 
$$\mathcal{H}=\mathcal{H}_1\times\mathcal{H}_2\times\mathcal{H}_3\times\mathcal{H}_4$$ be a suitable convex space of measurable functions with
$\mathcal{H}_1=\{h_1:(\theta,x)\mapsto\mathbb{R}\}$, $\mathcal{H}_2=\{h_2:\theta\mapsto\mathbb{R}^+\}$, $\mathcal{H}_3=\{h_3:(\theta,x)\mapsto\mathbb{R}\}$ and $\mathcal{H}_4=\{h_4:\theta\mapsto\mathbb{R}\}$. We obtain the moment function $$\psi\big((Y,X),\theta,h\big):\left(\mathcal{Y}\times\mathcal{X}\right)\times\Theta\times\mathcal{H}\to\mathbb{R}$$ with 
\begin{align*}
\big((Y,X),\theta,h\big)\mapsto&-\underbrace{\frac{h_4(\theta) }{2h_2(\theta)}}_{=:I(\theta,h_2,h_4)}-\underbrace{\frac{1}{h_2(\theta)}(\Lambda_{\theta}(Y)-h_1(\theta,X))(\dot\Lambda_{\theta}(Y)-h_3(\theta,X))}_{=:II(\theta,h_1,h_2,h_3)}\\
&+\underbrace{\frac{h_4(\theta)}{2\big(h_2(\theta)\big)^2}(\Lambda_{\theta}(Y)-h_1(\theta,X))^2}_{=:III(\theta,h_1,h_2,h_4)}+c_\theta,
\end{align*}
where $c_\theta:=\frac{\dot\Lambda'_{\theta}(Y)}{\Lambda'_{\theta}(Y)}$ and $\mathcal{Y}$ is the support of $Y$, respectively, $\mathcal{X}$ is the support of $X$.\\ \\
Under regularity conditions, the next lemma ensures the identification of the transformation parameter. The condition \ref{A1}, \ref{A5}, and \ref{A6}, as stated on page \textit{11}, are only sufficient conditions and they will be required for the main theorem in \ref{main ref}.
\begin{lemma}{}\label{identification}
Under the conditions \ref{A1}, \ref{A5}, and \ref{A6} the true parameter $\theta_0$ is identified as a unique solution of the moment condition
$$\E \Big[\psi\big((Y,X),\theta_0,h_0\big)\Big]=0.$$
\end{lemma}
\begin{proof}
We use the same argument as Neumeyer, Noh, Van Keilegom (2016) \cite{neumeyer2014heteroscedastic}.
Define
\begin{align*}
f^{(\theta)}(y|x):=\frac{1}{\sqrt{2\pi\sigma_\theta^2}}\exp\left(-\frac{(\Lambda_{\theta}(y)-x\beta_{\theta})^2}{2\sigma_\theta^2}\right)\Lambda_{\theta}'(y).
\end{align*}
The expected Kullback-Leibler-Distance between $f_{Y|X}$ and $f^{(\theta)}$ is greater or equal to zero and equality only holds for the true parameter $\theta_0$.
Therefore, the following expression is minimized in $\theta_0$
\begin{align*}
&\quad\int\int\log\left(\frac{f_{Y|X}(y|x)}{f^{(\theta)}(y|x))}\right)f_{Y|X}(y|x)dydF_X(x)\\ &=\underbrace{\int\int\log(f_{Y|X}(y|x)f_{Y|X}(y|x)dydF_X(x)}_{constant}\\&\quad-\underbrace{\int\int\log(f^{(\theta)}(y|x))f_{Y|X}(y|x)dydF_X(x)}_{=E[\log(f^{(\theta)}(Y|X))]}.
\end{align*}
It follows that $E[\log(f^{(\theta)}(Y|X))]$ is maximized for the true parameter $\theta=\theta_0$. Under regularity conditions \ref{A1}, \ref{A5} and \ref{A6}
\begin{align*}E \Big[\psi\big((Y,X),\theta_0,h_0\big)\Big]&=E\left[\frac{\partial}{\partial \theta}\log(f^{(\theta)}(Y|X))\big|_{\theta=\theta_0}\right]\\&=\frac{\partial}{\partial \theta} E[\log(f^{(\theta)}(Y|X))]\big|_{\theta=\theta_0}=0.
\end{align*}
Here we used that for all $\theta$
\begin{align*}
0<c\le\sigma_\theta^2\quad \text{and}\quad \sigma_\theta^2\le\E\left[\sup\limits_{\theta\in\Theta}\varepsilon^2_\theta\right]\le C<\infty, 
\end{align*}
which is shown in the proof of theorem \ref{Uniformlasso}.
\end{proof}
\newpage
\section{Main Results}\label{main section}
This section focuses on the central elements of our Z-estimation problem, which are Neyman orthogonality, uniform estimation of the nuisance function, and a theorem about the asymptotic distribution of the estimated transformation parameter based on an entropy condition. In the following, we consider the model described in section \ref{normal}.
\subsection{Neyman orthogonality condition}\label{orthogonality}
To be able to use plug-in estimators for the nuisance function, the moment condition to identify $\theta_0$ needs to be insensitive towards small changes in the estimated nuisance function. This property is granted by the Neyman orthogonality condition that is defined in Chernozhukov et al. (2017) \cite{chernozhukov2017double}. In this work, the authors describe the condition in great detail and they provide an extensive overview of the settings, where the condition holds.\\
To prove the Neyman orthogonality condition, we define the Gateaux derivative with respect to some $h\in\mathcal{H}$ in $h_0$
$$D_r[h-h_0]:=\partial_r\Big\{\E\Big[\psi\Big((Y,X),\theta_0,h_0+r(h-h_0)\Big)\Big]\Big\},$$
where 
\begin{align*}
&\quad h_0+r(h-h_0)\\&:=\Big(m_\theta +r(h_1-m_\theta),\sigma^2_\theta+r(h_2-\sigma^2_\theta),\dot m_\theta+r(h_1-\dot m_\theta),\dot\sigma^2_\theta+r(h_2-\sigma^2_\theta)\Big).
\end{align*}
It is important to mention that $\mathcal{H}_1$ to $\mathcal{H}_4$ are assumed to be convex, which is why the term $\psi\Big((Y,X),\theta_0,h_0+r(h-h_0)\Big)$ is well defined and exists for all $r\in [0,1)$.
\begin{lemma}{}\label{orthogonalityprop}
Let $\mathcal{H}'\subseteq \mathcal{H}$.
Under the conditions 
\begin{align*}
&\E\left[\sup\limits_{\theta\in\Theta}\varepsilon_\theta^2\right]<\infty\quad
\text{and}\quad \E\left[\sup\limits_{h\in\mathcal{H}'}\left|\psi\Big((Y,X),\theta_0,h\Big)\right|\right]<\infty
\end{align*}
the Neyman orthogonality condition
$$D_0[h-h_0]=0$$
is satisfied for all $h\in\mathcal{H}'$.
\end{lemma}
It is sufficient to restrict the condition onto the nuisance realization set, as defined in Section \ref{entropysec}, which contains the estimated nuisance function with probability $1-o(1)$.\\
Our estimation procedure is closely related to the 'concentrated out' approach in general Likelihood and other M-estimation problems as described in Chernozhukov et al. (2017) \cite{chernozhukov2017double} and Newey (1994) \cite{newey1994asymptotic}. In Lemma 2.5, Chernozhukov et al. (2017) \cite{chernozhukov2017double} give conditions when the score $\psi$ is Neyman orthogonal at $(\theta_0,h_0)$. They suppose that the target parameter $\theta$ and the nuisance parameter $h_0(\theta)$ solve optimization problem
\begin{align}
\max\limits_{\theta\in\Theta,h\in\mathcal{H}}\E[l((Y,X),\theta,h(\theta))]
\end{align}
and 
\begin{align}\label{db0}
h_0(\theta)=\arg\max\limits_{h\in\mathcal{H}}\E[l((Y,X),\theta,h(\theta))]
\end{align}
for all $\theta\in\Theta$, where $l$ is a known criterion function. However, our model does not fit in this setting because we set  
\begin{align*}
l_{Y|X}(\theta)&=-\frac{n}{2}\log(2\pi\sigma_\theta^2)-\frac{1}{2\sigma_\theta^2}\sum_{i=1}^n(\Lambda_{\theta}(Y_i)-m_{\theta}(X_i))^2+\sum_{i=1}^n\log(\Lambda_\theta'(Y_i)),
\end{align*} 
which is the log-likelihood of our model \ref{model} only if $\theta=\theta_0$. Therefore, in general, $h_0(\theta)$ does not satisfy \ref{db0} and our problem is not covered by this setting.\\ \\
Next, we give a set of assumptions that are needed for the following theorems and have already been used   for lemma \ref{identification} (\ref{A1}, \ref{A5}, and \ref{A6}).\\ \\
Assumptions \textbf{A1}-\textbf{A11}.\\
The following assumptions hold uniformly in $n\ge n_0,P\in\mathcal{P}_n$:
\begin{enumerate}[label=\textbf{A\arabic*},ref=A\arabic*]
\item\label{A1}
\begin{align*}
\E\left[\sup\limits_{\theta \in \Theta}|\log(\Lambda_\theta'(Y))|\right]<\infty
\end{align*}
\item \label{A3}
 The parameters obey the growth condition $$s\log(p\vee n)=o\left(n^{1/4}\right).$$
 \item \label{A4}
 For all $n\in\mathbb{N}$ the regressor $X=(X_1,\dots,X_p)$ has a bounded support $\mathcal{X}$.
\item \label{A5}
Uniformly, in $\theta$ the variance of the error term and its derivation with respect to the transformation parameter is bounded away from zero, namely
\begin{align*}
    0&<c\le \inf\limits_{\theta\in\Theta}\E\big[\varepsilon_\theta^2\big]\\
    0&<c\le \inf\limits_{\theta\in\Theta}\E\big[\dot\varepsilon_\theta^2\big].
    \end{align*}
\item \label{A6}
The transformations and its derivations are measurable and the classes of functions 
    $$\mathcal{F}_{\Lambda}:=\big\{\Lambda_{\theta}(\cdot)|\theta\in\Theta\big\}\quad \dot{\mathcal{F}}_{\Lambda}:=\big\{\dot\Lambda_{\theta}(\cdot)|\theta\in\Theta\big\}$$
    have VC index $C_{\Lambda}$, respectively, $\dot C_{\Lambda}$.\\
    The classes $\mathcal{F}_{\Lambda}$ and $\dot{\mathcal{F}}_{\Lambda}$ have envelopes $F_{\Lambda}$, respectively, $\dot F_{\Lambda}$ with 
    $$\E [F_{\Lambda}(Y)^{14}]<\infty\quad\text{and}\quad\E [\dot F_{\Lambda}(Y)^8]<\infty.$$
\item \label{A7}
The following condition for the second derivation of the transformation with respect to $\theta$ holds:
    $$\sup\limits_{\theta\in\Theta}\E \Big[\big(\ddot\Lambda_\theta(Y)\big)^2\Big]\le C.$$
\item \label{A8}
    With probability $1-o(1)$, the empirical minimum and maximum sparse eigenvalues are bounded from zero and above, namely
    \begin{align*}
		0<\kappa'&\le\inf\limits_{||\delta||_0\le s\log(n),||\delta||=1}||X^T\delta||_{\mathbb{P}_{n,2}}\\&\le\sup\limits_{||\delta||_0\le s\log(n),||\delta||=1}||X^T\delta||_{\mathbb{P}_{n,2}}\le\kappa''<\infty.
		\end{align*}
\item \label{A9}
The class of functions 
    $$\mathcal{J}_{\Lambda}:=\left\{c_\theta(\cdot)=\frac{\dot\Lambda_{\theta}'(\cdot)}{\Lambda_{\theta}'(\cdot)}\bigg|\theta\in\Theta\right\}$$
    has an envelope $J_{\Lambda}$ with 
    $$\E [J_{\Lambda}(Y)^6]<\infty.$$
\item \label{A10}
For all $\theta\in\Theta$, $\tilde{h}\in\tilde{\mathcal{H}}$ it holds that
\begin{itemize}
\item[(i)]
\begin{align*}
\E\Big[\Big(\psi\big((Y,X),\theta,h_0(\theta,X)\big)-\psi\big((Y,X),\theta_0,h_0(\theta_0,X)\big)\Big)^2\Big]\le C|\theta-\theta_0|^{2}
\end{align*}
\item[(ii)]
\begin{align*}
&\quad\E\Big[\Big(\psi\big((Y,X),\theta,\tilde{h}(\theta,X)\big)-\psi\big((Y,X),\theta,h_0(\theta,X)\big)\Big)^2\Big]\\&\le C\E\left[\|\tilde{h}(\theta,X)-h_0(\theta,X)\|^{2}_2\right]
\end{align*}
\item[(iii)]
\begin{align*}
&\quad\sup\limits_{r\in(0,1)}\left|\partial^2_r\bigg\{\E\Big[\psi\big((Y,X),\theta_0+r(\theta-\theta_0),h_0+r(\tilde{h}-h_0)\big)\Big]\bigg\}\right|\\
&\le C\left(|\theta-\theta_0|^2+\sup\limits_{\theta^*\in\Theta}\E\Big[\|\tilde{h}(\theta^*,X)-h_0(\theta^*,X)\|^{2}_2\Big]\right)
\end{align*}
\end{itemize}
for a constant $C$ independent from $\theta$ and $\tilde{\mathcal{H}}$ defined in \ref{entropysec}.\\
\item \label{A11}
For $h \in \tilde{\mathcal{H}}$ the function
$$\theta\mapsto \E \Big[\psi\big((Y,X),\theta,h(\theta,X)\big)\Big]$$
is differentiable in a neighborhood of $\theta_0$ and
for all $\theta\in\Theta$, the identification relation
$$2|\E[\psi((Y,X)),\theta,h_0(\theta,X)]|\ge |\Gamma(\theta-\theta_0)|\wedge c_0$$
is satisfied with
$$\Gamma:=\partial_\theta\E \Big[\psi\big((Y,X),\theta_0,h_0(\theta_0,X)\big)\Big]>c_1.$$
\item \label{A12}
The map $(\theta,h)\mapsto\E[\psi((X,Y),\theta,h)]$ is twice continuously Gateaux-differentiable on $\Theta\times\mathcal{H}$.
\end{enumerate}\ \\
Assumptions \ref{A1}-\ref{A12} are a set of sufficient conditions for the main result stated in theorem \ref{main}. Assumption \ref{A1} allows us to interchange derivation and integration, which is necessary for the verification of the moment condition. The growth condition \ref{A3} restricts the sparsity index. Assumptions \ref{A3}, \ref{A7}, and \ref{A8} are needed for the uniform estimation of the nuisance function. Condition \ref{A3} can be relaxed, cf assumption 6.1 from Belloni et al. (2017) \cite{belloni2017program}. Assumption \ref{A8} is a standard eigenvalue condition for the lasso estimation. Condition \ref{A5} prevents degenerate distributions in the models \ref{model} and \ref{modeldot}. Assumptions \ref{A6} and \ref{A9} control the complexity of the class of transformations and bound the moments uniformly over $\theta$. Assumption \ref{A10} is a set of mild smoothness conditions. Assumption \ref{A11} implies sufficient identifiability of the true transformation parameter $\theta_0$. Assumption \ref{A12} only requires differentiability of the function $(\theta,h)\mapsto\E[\psi((X,Y),\theta,h)]$, which is a weaker condition than the differentiability of the function $(\theta,h)\mapsto\psi((X,Y),\theta,h)$.\\ \\
The following lemma shows that the first part of condition \ref{A6} is satisfied for the popular Box-Cox power transformations and the modification proposed by Yeo and Johnson.
\begin{lemma}\label{VC}
The class of Box-Cox-Transformations $\mathcal{F}_1=\{\Lambda_{\theta}(\cdot)|\theta\in\mathbb{R}\}$ and the class of derivatives $\mathcal{F}_2=\{\dot\Lambda_{\theta}(\cdot)|\theta\in\mathbb{R}\}$ with respect to the transformation parameter $\theta$ are VC classes. The same holds for Yeo-Johnson power transformations.
\end{lemma}
The proof of the lemma is given in the appendix.
\subsection{Uniform estimation of the nuisance functions}\label{uniformsection}
The rates for the estimation of the regression functions $m_\theta$ and $\dot m_\theta$ can be directly obtained by the uniform prediction rates of the lasso estimator.
\begin{theorem}\label{Uniformlasso1}\ \\
Under assumptions \ref{A1}-\ref{A8}, uniformly for all $P\in\mathcal{P}_n$ with probability $1-o(1)$, it holds that:
\begin{align}
\sup\limits_{\theta\in\Theta}||\hat{\beta}_\theta||_0&=O(s)\\
\sup\limits_{\theta\in\Theta}||X^T(\hat{\beta}_\theta-\beta_\theta)||_{\mathbb{P}_{n,2}}&=o\left(n^{-\frac{1}{4}}\right)\\
\sup\limits_{\theta\in\Theta}||\hat{\beta}_\theta-\beta_\theta||_{1}&=o\left(n^{-\frac{1}{4}}\right).
\end{align}
Respectively
\begin{align}
\sup\limits_{\theta\in\Theta}||\hat{\dot \beta}_\theta||_0&=O(s)\\
\sup\limits_{\theta\in\Theta}||X^T(\hat{\dot \beta}_\theta-\dot \beta_\theta)||_{\mathbb{P}_{n,2}}&=o\left(n^{-\frac{1}{4}}\right)\\
\sup\limits_{\theta\in\Theta}||\hat{\dot\beta}_\theta-\dot\beta_\theta||_{1}&=o\left(n^{-\frac{1}{4}}\right).
\end{align}
\end{theorem}\ \\
As a consequence of the uniform rates of the Lasso estimator, we are able to achieve uniform rates for the estimation of the variance $\sigma_\theta^2$ and its derivation $\dot \sigma_\theta^2$.
\begin{theorem}\label{varianceest}\ \\
Under assumptions of Theorem \ref{Uniformlasso1}, uniformly for all $P\in\mathcal{P}_n$ with probability $1-o(1)$, it holds that:
\begin{align}\sup\limits_{\theta\in\Theta}|\hat{\sigma}^2_{\theta}-\sigma^2_{\theta}|&=o\left(n^{-\frac{1}{4}}\right)\\
\sup\limits_{\theta\in\Theta}|\hat{\dot\sigma}^2_{\theta}-\dot\sigma^2_{\theta}|&=o\left(n^{-\frac{1}{4}}\right).
\end{align}
\end{theorem}\ \\
\subsection{Entropy condition}\label{entropysec}
At first we define the following classes of functions 
\begin{align*}
\tilde{\mathcal{H}}_1&:=\left\{\tilde{h}_1:\Theta\times \mathcal{X}\to \mathbb{R}|\ \tilde{h}_1(\theta,x)=x\tilde{\beta}_\theta,\|\tilde{\beta}_\theta\|_0\le Cs, \|\tilde{\beta}_\theta-\beta_\theta\|_1\le Cn^{-\frac{1}{4}}\right\}\\
\tilde{\mathcal{H}}_2&:=\left\{\tilde{h}_2:\Theta\to \mathbb{R}^+|\ |\tilde{h}_2(\theta)-\sigma^2_\theta|\le Cn^{-\frac{1}{4}}\right\}\\
\tilde{\mathcal{H}}_3&:=\left\{\tilde{h}_3:\Theta\times \mathcal{X}\to \mathbb{R}|\ \tilde{h}_3(\theta,x)=x\tilde{\beta}_\theta,\|\tilde{\beta}_\theta\|_0\le Cs, \|\tilde{\beta}_\theta-\dot{\beta}_\theta\|_1\le Cn^{-\frac{1}{4}}\right\}\\
\tilde{\mathcal{H}}_4&:=\left\{\tilde{h}_4:\Theta\to \mathbb{R}|\ |\tilde{h}_4(\theta)-\dot\sigma^2_\theta|\le Cn^{-\frac{1}{4}}\right\}
\end{align*}
and 
$$\tilde{\mathcal{H}}:=\tilde{\mathcal{H}}_1\times \tilde{\mathcal{H}}_2\times \tilde{\mathcal{H}}_3\times \tilde{\mathcal{H}}_4.$$
The set $\tilde{\mathcal{H}}$ is called the nuisance realization set. Theorems \ref{Uniformlasso1} and \ref{varianceest} enable us to choose constants $C$ independent from $\theta$ and still contain the estimated functions in $\tilde{\mathcal{H}}$ with probability $1-o(1)$.\\
Furthermore, for an arbitrary but fixed $\theta\in\Theta$  we define the following projections 
\begin{align*}
\tilde{\mathcal{H}}_1(\theta)&:=\left\{\tilde{h}_1:\mathcal{X}\to \mathbb{R}|\ \tilde{h}_1(x)=\tilde{h}_1(\theta,x)\in\tilde{\mathcal{H}}_1\right\}\\
\tilde{\mathcal{H}}_2(\theta)&:=\left\{c\in\mathbb{R}^+\big|\ |c-\sigma_\theta^2|\le Cn^{-1/4}\right\}
\end{align*}
and $\tilde{\mathcal{H}}_3(\theta)$, $\tilde{\mathcal{H}}_4(\theta)$, respectively, $\tilde{\mathcal{H}}(\theta)$, analogously.\\
We restrict the entropy of $\tilde{\mathcal{H}}(\theta)$ uniformly over $\theta$ to use the maximal inequality stated in Theorem 5.1 from Chernozhukov et al. (2014) \cite{chernozhukov2014gaussian}. This enables us to bound the empirical process in the proof of theorem \ref{generalmain} Step 1.
\begin{theorem}{}\label{entropy}
Under the assumptions \ref{A5}, \ref{A6}, and \ref{A9} the class of functions 
$$\Psi(\theta)=\Big\{(y,x)\mapsto \psi\big((y,x),\theta,\tilde{h}(\theta,x)\big),\tilde{h}\in\tilde{\mathcal{H}}(\theta)\Big\}$$
has a measurable envelope $\bar\psi\ge \sup_{\psi\in \Psi(\theta)}|\psi|$ independent from $\theta$, such that for some $q\ge 4$
$$\E \Big[(\bar\psi(Y,X))^q\Big]\le C_1.$$
The class $\Psi(\theta)$ is pointwise measurable and uniformly for all $\theta\in\Theta$
$$\sup_{Q}\log N(\varepsilon||\bar\psi||_{Q,2},\Psi(\theta),L_2(Q))\le C_1 s\log\left(\frac{C_2(p\vee n)}{\varepsilon}\right)$$
with $C_1$ and $C_2$ being independent from $\theta$.
\end{theorem}
The motivation of the entropy condition stated in theorem \ref{entropy} is described in \ref{entropycomment}. The entropy condition and the results in subsection \ref{orthogonality} and subsection \ref{uniformsection} enable us to establish the asymptotic distribution of our estimated transformation parameter.
\subsection{Main Theorem}\label{main ref}
The main theorem provides that our estimator $\hat{\theta}$ converge with rate $1/\sqrt{n}$ and is asymptotic unbiased and normal. 
\begin{theorem}{} \label{main}
Under the assumptions \ref{A1}-\ref{A12} the estimator $\hat{\theta}$ in (\ref{estimator}) obeys
$$n^{\frac{1}{2}}(\hat{\theta}-\theta_0)\xrightarrow{\mathcal{D}}\mathcal{N}(0,\Sigma),$$
where $$\Sigma:=\E\Big[\Gamma^{-2}\psi^2\big((Y,X),\theta_0,h_0(\theta_0,X)\big)\Big]$$ with $\Gamma=\partial_\theta\E \Big[\psi\big((Y,X),\theta_0,h_0(\theta_0,X)\big)\Big]$.
\end{theorem}
A standard bootstrap can be applied to estimate the unknown variance $\Sigma$. Therefore, asymptotic level-$\alpha$ tests for null hypothesis can be constructed from theorem \ref{main}.


\section{Simulation}\label{simulation}

This section provides a simulation study of the proposed estimator. The data generating process is given by
\begin{align*}
\Lambda_{\theta_0}(Y)=X\beta_{\theta_0}+\varepsilon.
\end{align*}
The coefficients are set to $$\beta_{\theta_0,j}=\begin{cases}1 \quad\text{for }j\le s\\
0\quad\text{for }j>s.\end{cases}$$
Therefore, $\beta_{\theta_0}$ is a sparse vector with $\|\beta_{\theta_0}\|_0=s$ and $\{\Lambda_\theta:\theta\in\Theta\}$ is a given class of transformations.
The design matrix is simulated as $$X\sim\mathcal{N}\left(0,\Sigma^{(X)}\right)$$
for the three different correlation structures
$$\Sigma_0^{(X)}=I_p,$$ 
$$\Sigma_1^{(X)}=(c^{|i-j|})_{i,j\in\{1,\dots,p\}}$$
and
$$\Sigma_2^{(X)}=(1-c^p)I_p+(c^{p-|i-j|})_{i,j\in\{1,\dots,p\}}$$
with $c=0.35$. The error terms $\varepsilon\sim\mathcal{N}(0,\sigma^2)$ iid, where $\sigma^2$ is chosen according to the correlation matrix $\Sigma^{(X)}$ to keep the signal-to-noise ratio (SNR) at a fixed level. To obtain the simulated values for $Y_i$, which are used for the estimation of $\theta_0$, we apply the inverse transformation $\Lambda_{\theta_0}^{-1}$ onto the simulated values of $\Lambda_{\theta_0}(Y_i)$.\\
We consider different classes of transformations, correlation structures, and we vary the number of the regressors $p$ as well as the SNR. The SNR is defined as
$$SNR=\frac{Var(X\beta_0)}{Var(\varepsilon)}.$$
The sample size is $n=200$ and the sparsity index is $s=5$. The number of repetitions is set to $R=500$. The accuracy of the estimate $\hat{\theta}$ is measured by the mean-absolute-error (MAE) 
$$\text{MAE}=\frac{1}{R}\sum\limits_{h=1}^R|\hat{\theta}_h-\theta_0|.$$
The accuracy of the predictive performance is measured out-of-sample on an independent sample (testing sample) by the mean-squared-error  
$$\text{MSE}=\frac{1}{R}\sum\limits_{h=1}^{R}\E_{n_t}\left[\left(\Lambda_{\theta_0}(Y)-X\hat{\beta}_{\hat{\theta}_h}\right)^2\right].$$
The empirical expectation $\E_{n_t}[\cdot]$ is taken over by a new and independent sample of size $n_t=200$. Both measures MAE and MSE are based on the unknown transformation parameter $\theta_0$.\\
Additionally, for a fixed level $\alpha=0.05$, we validate the significance level (acceptance rate) of a test of the form 
$$H_0: \theta_0=\theta.$$
We test if a given $\theta\in\Theta$ is the right transformation parameter to guarantee normal distributed errors. Therefore, we estimate the unknown variance $\Sigma$ via bootstrap by drawing $k=100$ bootstrap samples and construct a $(1-\alpha)$-confidence interval of the form
$$\left[\theta-\sqrt{\hat{\Sigma}}z_{(1-\alpha/2)}, \theta+\sqrt{\hat{\Sigma}}z_{(1-\alpha/2)}\right],$$
where $z_\gamma$ is the $\gamma$-quantile of the standard normal distribution. The empirical rejection rate is reported.\\
\subsection{Box-Cox-Transformations}
In the first setting, we analyze the class of Box-Cox-Transformations. The Box-Cox-Transformations are defined as 
\begin{align*}
\Lambda_\theta(y)=\begin{cases} \frac{y^{\theta}-1}{\theta}&\text{for $\theta\neq 0$}\\ \log(y)&\text{for $\theta=0$.}\end{cases}
\end{align*}
This class and the class of its derivatives with respect to the transformation parameter $\theta$ are VC classes by lemma \ref{VC}.
To ensure that we simulate non-negative observations, we restrict our simulation to the true transformation parameter $\theta_0=0$ and vary the number of regressors from $20$ to $100$. The results for the various settings are summarized in Tables \ref{bxcx1}--\ref{bxcx3}. 
In all three settings, the average of the estimator is close to the true value of $0$. The acceptance rate and relative MSE seem to be comparable for all three settings, while the second setting has the smallest MAE.
In summary, the results reveal that the estimated parameter value is, on average, close to the true one and that the the acceptance rate is close to the nominal level of $95\%$. \\ \\
  
\begin{table}[ht]
\centering
\begin{tabular}{cccccccc}
  \hline
p & n & $\theta_0$ & SNR & Estimator & Acceptance rate & MAE & rel. MSE \\ 
  \hline
20 & 200 & 0.0 & 1.0 & -0.00116535 & 0.946 & 0.0134 & 0.9801 \\ 
  20 & 200 & 0.0 & 3.0 & -0.00069033 & 0.954 & 0.0135 & 0.9805 \\ 
  50 & 200 & 0.0 & 1.0 & 0.00043417 & 0.946 & 0.0138 & 1.0529 \\ 
  50 & 200 & 0.0 & 3.0 & 0.00049114 & 0.938 & 0.0139 & 1.0535 \\ 
   \hline
\end{tabular}
\caption{Simulation for $\Sigma^{(X)}=I_p$}
\label{bxcx1} 
\end{table}

\begin{table}[ht]
\centering
\begin{tabular}{cccccccc}
  \hline
p & n & $\theta_0$ & SNR & Estimator & Acceptance rate & MAE & rel. MSE \\ 
  \hline
20 & 200 & 0.0 & 1.0 & -0.00087619 & 0.944 & 0.0102 & 0.9789 \\ 
  20 & 200 & 0.0 & 3.0 & -0.00049946 & 0.952 & 0.0103 & 0.9796 \\ 
  50 & 200 & 0.0 & 1.0 & 0.00026786 & 0.956 & 0.0104 & 1.0438 \\ 
  50 & 200 & 0.0 & 3.0 & 0.00033492 & 0.934 & 0.0105 & 1.0439 \\ 
   \hline
\end{tabular}
\caption{Simulation for $\Sigma^{(X)}=\Sigma_1^{(X)}$}
\label{bxcx2}  
\end{table}

\begin{table}[ht]
\centering
\begin{tabular}{cccccccc}
  \hline
p & n & $\theta_0$ & SNR & Estimator & Acceptance rate & MAE & rel. MSE \\ 
  \hline
20 & 200 & 0.0 & 1.0 & -0.00118608 & 0.952 & 0.0135 & 0.9836 \\ 
  20 & 200 & 0.0 & 3.0 & -0.00070107 & 0.954 & 0.0136 & 0.9842 \\ 
  50 & 200 & 0.0 & 1.0 & 0.00065337 & 0.950 & 0.0139 & 1.0532 \\ 
  50 & 200 & 0.0 & 3.0 & 0.00069861 & 0.956 & 0.0139 & 1.0543 \\ 
   \hline
\end{tabular}
\caption{Simulation for $\Sigma^{(X)}=\Sigma_2^{(X)}$}
\label{bxcx3}  
\end{table}
 
\begin{figure}[H]
 \label{fig1}
 \caption{Empirical distribution of the estimator}
 \centering
\begin{knitrout}
\definecolor{shadecolor}{rgb}{0.969, 0.969, 0.969}\color{fgcolor}
\includegraphics[width=\textwidth]{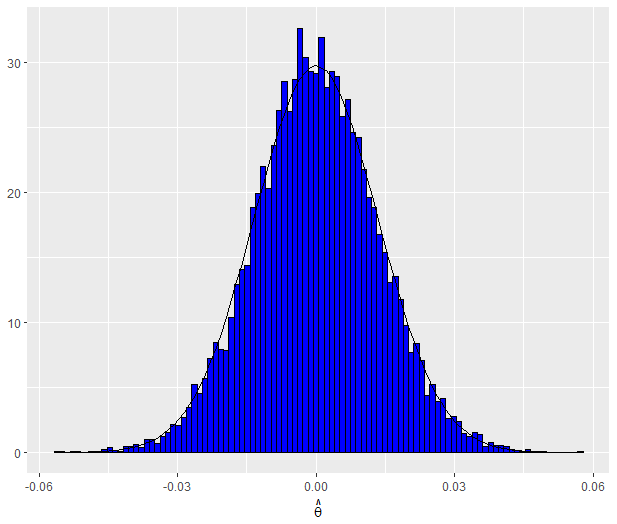} 

\end{knitrout}
\label{fignorm}
\end{figure}
Figure \ref{fignorm} shows the empirical distribution of $\hat{\theta}$ generated by $10000$ independent simulations from the last setting in Table \ref{bxcx2} with SNR=$1$. This confirms that our estimator is normally distributed.

\subsection{Yeo-Johnson Power Transformations}
Next, we consider the class of Yeo-Johnson Power Transformations. The Yeo-Johnson Power Transformations are defined as
\begin{align*}
\Lambda_\theta(y) = \begin{cases} \frac{(y+1)^\theta-1}{\theta}, &\text{for } y \ge 0,\theta\neq 0 \\
 \log(y+1), &\text{for } y\ge 0,\theta=0\\
 -\frac{(-y+1)^{2-\theta}-1}{2-\theta}, &\text{for } y<0,\theta\neq 2 \\
 -\log(-y+1), &\text{for } y<0,\theta=2.
 \end{cases}
\end{align*}
We set the true transformation parameter $\theta_0=1$ and vary the number of regressors from $20$ to $100$, again. The results are summarized in Tables \ref{yj1}--\ref{yj3}. We get similar patterns as under the Box-Cox transformation. We get similar results for acceptance rate, average estimator, and relative MSE for all three settings. The MAE of the second setting is smallest.

We can summarize that the test is close to the nominal level of $95\%$ and the transformation parameter is estimated accurately.\\

\begin{table}[ht]
\centering
\begin{tabular}{cccccccc}
  \hline
p & n & $\theta_0$ & SNR & Estimator & Acceptance rate & MAE & rel. MSE \\ 
  \hline
20 & 200 & 1.0 & 1.0 & 0.99797054 & 0.952 & 0.0339 & 0.9811 \\ 
  20 & 200 & 1.0 & 3.0 & 0.99905333 & 0.960 & 0.0310 & 0.9812 \\ 
  50 & 200 & 1.0 & 1.0 & 1.00162067 & 0.954 & 0.0350 & 1.0536 \\ 
  50 & 200 & 1.0 & 3.0 & 1.00134004 & 0.936 & 0.0325 & 1.0543 \\ 
   \hline
\end{tabular}
\caption{Simulation for $\Sigma^{(X)}=I_p$}
\label{yj1}  
\end{table}

\begin{table}[ht]
\centering
\begin{tabular}{cccccccc}
  \hline
p & n & $\theta_0$ & SNR & Estimator & Acceptance rate & MAE & rel. MSE \\ 
  \hline
20 & 200 & 1.0 & 1.0 & 0.99833234 & 0.954 & 0.0296 & 0.9795 \\ 
  20 & 200 & 1.0 & 3.0 & 0.99916908 & 0.960 & 0.0270 & 0.9802 \\ 
  50 & 200 & 1.0 & 1.0 & 1.00108300 & 0.944 & 0.0304 & 1.0438 \\ 
  50 & 200 & 1.0 & 3.0 & 1.00121914 & 0.934 & 0.0281 & 1.0450 \\ 
   \hline
\end{tabular}
\caption{Simulation for $\Sigma^{(X)}=\Sigma_1^{(X)}$}
\label{yj2} 
\end{table}

\begin{table}[ht]
\centering
\begin{tabular}{cccccccc}
  \hline
p & n & $\theta_0$ & SNR & Estimator & Acceptance rate & MAE & rel. MSE \\ 
  \hline
20 & 200 & 1.0 & 1.0 & 0.99783058 & 0.952 & 0.0341 & 0.9842 \\ 
  20 & 200 & 1.0 & 3.0 & 0.99895182 & 0.950 & 0.0313 & 0.9850 \\ 
  50 & 200 & 1.0 & 1.0 & 1.00180567 & 0.948 & 0.0351 & 1.0539 \\ 
  50 & 200 & 1.0 & 3.0 & 1.00181213 & 0.942 & 0.0327 & 1.0552 \\ 
   \hline
\end{tabular}
\caption{Simulation for $\Sigma^{(X)}=\Sigma_2^{(X)}$}
\label{yj3} 
\end{table}

\newpage

\section{Application}\label{application}
\subsection{Econometric Specification of the Wage Equation}\label{wagemodel}
In labor economics, the analysis of wage data is key. In addition, labor economics aims to identify the determinants of wages, to estimate a so-called Mincer equation, and evaluate the impact of labor market programs on wages. Wages are non-negative and show a high-degree of skewness, which is not compatible with a normal distribution. Hence, wages are transformed in almost all studies by the logarithm. Figure \ref{figwage} shows the weekly wage distribution from the US survey data, which are described in the next section.  Here, we focus on the estimation of a Mincer type equation. The Mincer equation formulates a relationship between log wages (W) and schooling (S), experience (Exp), and other control variables (X, p-dimensional):
\begin{align}
\log W = \alpha + \beta S + \gamma Exp + \delta Exp^2 + \mu' X +\varepsilon
\end{align}
with $\varepsilon\sim\mathcal{N}(0,\sigma^2)$. $\alpha, \beta, \gamma, \delta$ are coefficients and $\mu$ is a p-dimensional vector of the coefficients of the control variables. We consider a high-dimensional setting where the set of potential control variables is high and we do not take it for granted that the log-transformation is appropriate but instead estimate a transformation model and test for the transformation parameter. The model is given by
\begin{align}
\Lambda_{\theta_0}(W)= \alpha_{\theta_0} + \beta S_{\theta_0} + \gamma Exp_{\theta_0} + \delta Exp_{\theta_0}^2 + \mu_{\theta_0}' X +\varepsilon_{\theta_0}
\end{align}
with $\varepsilon_{\theta_0}\sim\mathcal{N}(0,\sigma^2)$.
\pagebreak

\begin{figure}[H]
 \label{fig2}
 \caption{Empirical wage distribution from the US survey data}
 \centering
\begin{knitrout}
\definecolor{shadecolor}{rgb}{0.969, 0.969, 0.969}\color{fgcolor}
\includegraphics[width=\maxwidth]{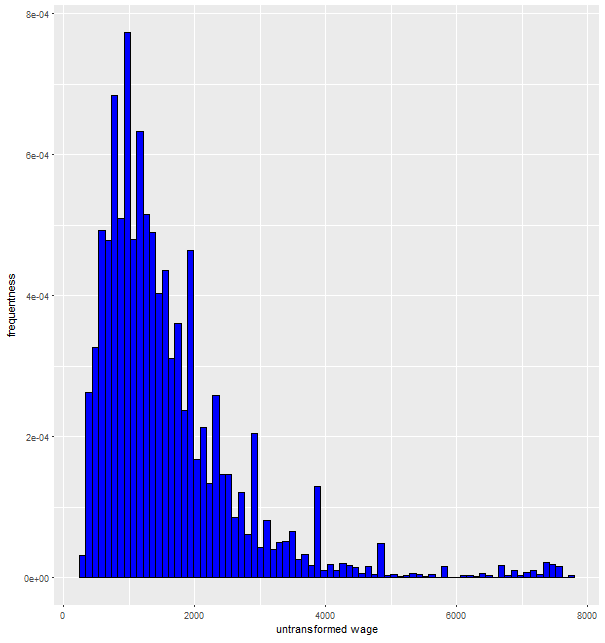} 

\end{knitrout}
\label{figwage}
\end{figure}

\subsection{Dataset}
\subsubsection{Overview}
In our empirical study, we use data from the 2015 American Community Survey (ACS), as provided by \cite{ipums} and extracted from the IPUMS-USA website\footnote{\url{https://usa.ipums.org/usa/}}. The ACS provides a 1\%-sample of the US population with mandatory participation. The data offers a large number of socio-economic characteristics at the individual and household level, such as education, industry, occupation, and earnings.\\
We restrict our attention to individuals who graduated from university and are working full time ($30+$ hours), at least $50$ weeks a year. Weekly earnings are computed as annual earnings divided by 52 (weeks). We exclude individuals with experience $> 60$ and age $>65$. Moreover, we discard individuals with a weekly wage of less than \$$10$ (which is likely to be unreasonable given that we only consider full-time employees). Then, we drop all observations with a weekly wage under the $2.5\%$-quantile and over the $97.5\%$-quantile. Our final sample comprises $315,291$  individual observations.\\
In our analysis, we use 14 initial regressors, which are either directly available from the ACS data or have been constructed. We list the variables in Table \ref{listofvars}. Mostly, we use the categories as provided in the ACS data. This might be particularly informative for the region, occupation, and industry variables for which different definitions exist. Moreover, we construct the variables 'years of education' and 'labor market experience' from the information available. We construct all of the two-way interactions of the initial regressors, where the categorical variables are transformed to level-wise dummies. Additionally, we include the variable 'field of degree' to account for the individual's educational background. Finally, we drop all of the constructed variables that are nearly constant over all observations and we end with a high-dimensional setting with a total of $1,743$ regressors.

\begin{table}[!ht]
\centering
\begin{tabular}{lrrrr}
  \hline 	  \\[0.8ex]
  Variable & Type &  Baseline Category  \\[0.8ex]
  \hline  \\[-0.8ex]
Female & binary &    \\ 
Marital status & six categories  & never married, single \\
Race & four categories &   White \\ 
English language skills & five categories &  speaks only English \\ 
Hispanic & binary &  \\
Veteran Status & binary  & \\ 
Industry & 14 categories  & wholesale trade\\
Occupation & 26 categories  & management, science, arts \\
Region (US census) & nine categories & New England division \\ 
Experience (years) & continuous & \\ 
Experience squared & continuous & \\
Years of Education & continuous & \\
Family Size & continuous &\\
Number of own young children & continuous & \\
Field of degree & 37 categories & administration, teaching \\
   \hline \\[0.8ex]
\end{tabular} 
\caption{List of Regressors} 
\label{listofvars}
\end{table}

\subsubsection{Descriptive Statistics}
Table \ref{summarystats} provides the summary statistics for a selection of the variables that are available in our final sample from the ACS data.  Figure $\ref{figwage}$ and the following descriptive statistics illustrate that the mean of weekly wage for university graduates is higher  than the median; hence, we have skewed data. Weekly wage is characterized by non-negativity and a high variability.\\
\begin{table}[!ht]
\centering
\begin{tabular}{lrrr}
  \hline 	  \\[0.8ex]
  Variable   & Mean & SD & Median \\[0.8ex]
  \hline  \\[-0.8ex]
Weekly wage & 1591.22 & 1100.20 & 1307.69\\ 
Experience (years) &  20.75 & 11.36& 21 \\ 
Years of Education &  16.91 & 1.24& 17\\ 
Female & 0.48 & 0.50 & - \\ 
White &  0.84 & 0.37 & -\\ 
Black/Negro &  0.07 & 0.25 & - \\ 
Chinese &  0.02 & 0.15 & -\\
Hispanic &  0.05 & 0.23 & -\\ 
Veteran Status &  0.05 & 0.21& -\\ 
 \hline \\[-0.8ex]
Sample Size & 315291 & & \\ 
   \hline \\[0.8ex]
\end{tabular} 
\caption{Summary Statistics, ACS Data} 
\label{summarystats}
\end{table}
\ \\
\subsection{Results}
The estimated transformation parameter is $\hat{\theta}=-0.1260646$. Because the confidence interval $[-0.1307524 ,-0.1213768 ]$ is based on the asymptotic normality of the estimate $\hat{\theta}$ and variance estimation via $300$ Bootstrap samples, we can reject the null hypothesis $\theta=0$ on a $5\%$ significance level, which is equivalent to a log transformation. In Figure \ref{quantiles}, we compare the Q-Q plot of the untransformed wages with the Q-Q plot of the transformed wages with our estimated parameter (with a normal distribution determined by the sample mean and sample variance) and with the Q-Q plot under log-transformation ($\theta=0$). The estimated errors without transformation are not normally distributed, whereas after the transformation of the response variable with $\hat{\theta}$ the estimated error term seem to fit a normal distribution quite well. Considering the transformation function, one can recognize that for a transformation parameter below zero the transformation has a stronger curvature (see figure \ref{transformationfunctions}). This implies that the wages are more positively skewed towards normal distribution after a log transformation. Although we reject the log-transformation, Figure \ref{quantiles} reveals that the log-transformation might give a reasonable approximation for applications in labor economics.\\

\begin{figure}[H]
 \caption{Comparison of the Q-Q plots}
 \centering
\begin{knitrout}
\definecolor{shadecolor}{rgb}{0.969, 0.969, 0.969}\color{fgcolor}
\includegraphics[width=\maxwidth]{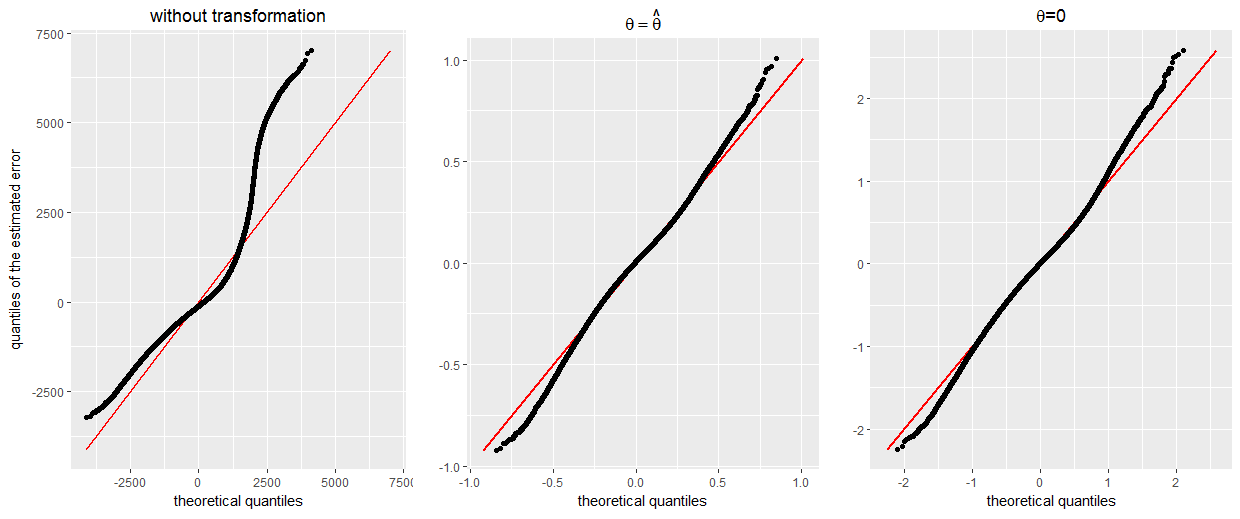} 
\end{knitrout}
\label{quantiles}
\end{figure}
 
\begin{figure}[H]
 \caption{Transformation function for $\theta=0$ (black) and $\theta=\hat{\theta}$ (red)}
 \centering
\begin{knitrout}
\definecolor{shadecolor}{rgb}{0.969, 0.969, 0.969}\color{fgcolor}
\includegraphics[width=\maxwidth]{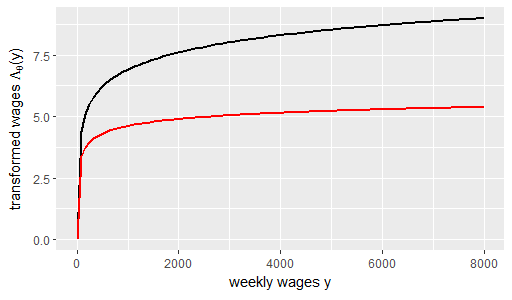} 
\end{knitrout}
\label{transformationfunctions}
\end{figure}

\pagebreak
\section{Conclusion}
In this paper, we propose an estimator for the transformation parameter in a high-dimensional setting. Transformation models, in particular the Box-Cox and Yeo-Johnson transformation, are very popular in applied statistics and econometrics. The rise of digitization has led to an increased availability of high-dimensional datasets and, hence, make it necessary to extend models for this setting when the number of variables $p$ is large (or even larger) compared to the sample size $n$. We build on the recent results on the Neyman orthogonality condition to prove the asymptotic normality of our estimator. The nuisance functions are estimated with lasso.

Our setting fits into a general Z-estimation problem with a high-dimensional nuisance function, which depends on the target parameter $\theta$. We extend the results in Belloni et al. (2014) \cite{belloni2014uniform} and Chernozhukov et al. (2017) \cite{chernozhukov2017double} to allow for an explicit dependence of the nuisance function on the target parameter $\theta$. This result might be of  interest for Z-estimation problems with the same structure.

In labor economics, wage is by default transformed with the logarithm. In our application, by analyzing US survey data we are able to show that the log-transformation is rejected on $5\%$ significance level but the log transformation might give an appropriate approximation.

In future research, we would like to address the problem of estimation and inference on elements of the coefficient vector of the regressors.

\newpage
\appendix
\section{Proofs}\label{appendixA}
\begin{proof}[Proof of Lemma \ref{orthogonalityprop}]\ \\
For the notation we refer to section \ref{identificationsec}. Let $h=(h_1,h_2,h_3,h_4)\in\mathcal{H}'$ arbitrary. First, we consider
\begin{align*}
&\quad\partial_r\Big(\Lambda_{\theta_0}(Y)-\big(m_{\theta_0}(X)+r\big(h_1(\theta_0,X)-m_{\theta_0}(X)\big)\big)\Big)\Big|_{r=0}\\&=m_{\theta_0}(X)-h_1(\theta_0,X)
\end{align*}
and analogous
\begin{align*}
&\quad\partial_r \Big(\dot \Lambda_{\theta_0}(Y)-\big(\dot m_{\theta_0}(X)+r\big(h_3(\theta_0,X)-\dot m_{\theta_0}(X)\big)\big)\Big)\Big|_{r=0}\\&=\dot m_{\theta_0}(X)-h_3(\theta_0,X).
\end{align*}
Additionally we have
$$\partial_r \left(\left(\sigma^2_{\theta_0}+r(h_2(\theta_0)-\sigma^2_{\theta_0})\right)^{-1}\right)\Big|_{r=0}=-\frac{h_2(\theta_0)-\sigma^2_{\theta_0}}{\left(\sigma^2_{\theta_0}\right)^2}$$
and 
$$\partial_r \left(\dot \sigma^2_{\theta_0}+r(h_4(\theta_0)-\dot\sigma^2_{\theta_0})\right)\Big|_{r=0}=h_4(\theta_0)-\dot\sigma^2_{\theta_0}.$$
By the product rule we obtain
\begin{align*}
&\quad\E\left[\partial_rI(\theta_0,\sigma^2_\theta+r(h_2-\sigma^2_\theta),\dot\sigma^2_\theta+r(h_4-\dot\sigma^2_\theta))|_{r=0}|X\right]\\&=\E \left[\frac{h_4(\theta_0)-\dot\sigma^2_{\theta_0}}{2\sigma^2_\theta}-\dot\sigma^2_\theta\frac{h_2(\theta_0)-\sigma^2_{\theta_0}}{2\left(\sigma^2_{\theta_0}\right)^2}\bigg|X\right]\\&=\frac{h_4(\theta_0)-\dot\sigma^2_{\theta_0}}{2\sigma^2_\theta}-\dot\sigma^2_\theta\frac{h_2(\theta_0)-\sigma^2_{\theta_0}}{2\left(\sigma^2_{\theta_0}\right)^2},
\end{align*}

\begin{align*}
&\E\left[\partial_rII(\theta_0,m_\theta +r(h_1-m_\theta),\sigma^2_\theta+r(h_2-\sigma^2_\theta),\dot m_\theta+r(h_1-\dot m_\theta))|_{r=0}|X\right]\\
=&\E \left[-\frac{h_2(\theta_0)-\sigma^2_{\theta_0}}{\left(\sigma^2_{\theta_0}\right)^2}\big(\Lambda_{\theta_0}(Y)-m_{\theta_0}(X)\big)\big(\dot\Lambda_{\theta_0}(Y)-\dot m_{\theta_0}(X)\big)\bigg|X\right]\\
&+\E \left[\frac{1}{\sigma_{\theta_0}^2}\big(m_{\theta_0}(X)-h_1(\theta_0,X)\big)\big(\dot\Lambda_{\theta_0}(Y)-\dot m_{\theta_0}(X)\big)\bigg|X\right]\\
&+\E \left[\frac{1}{\sigma_{\theta_0}^2}\big(\Lambda_{\theta_0}(Y)-m_{\theta_0}(X)\big)\big(\dot m_{\theta_0}(X)-h_3(\theta_0,X)\big)\bigg|X\right]\\
=&-\frac{h_2(\theta_0)-\sigma^2_{\theta_0}}{\left(\sigma^2_{\theta_0}\right)^2}\E \left[\varepsilon_{\theta_0}\dot \varepsilon_{\theta_0}\big|X\right]+\frac{m_{\theta_0}(X)-h_1(\theta_0,X)}{\sigma_{\theta_0}^2}\underbrace{\E \left[\dot \varepsilon_{\theta_0}\big|X\right]}_{=0}\\&+\frac{\dot m_{\theta_0}(X)-h_3(\theta_0,X)}{\sigma_{\theta_0}^2}\underbrace{\E \left[\varepsilon_{\theta_0}\big|X\right]}_{=0}\\
=&-\frac{h_2(\theta_0)-\sigma^2_{\theta_0}}{\left(\sigma^2_{\theta_0}\right)^2}\E \left[\varepsilon_{\theta_0}\dot \varepsilon_{\theta_0}\big|X\right],
\end{align*}
and
\begin{align*}
&\E\left[\partial_rIII(\theta_0,m_\theta +r(h_1-m_\theta),\sigma^2_\theta+r(h_2-\sigma^2_\theta),\dot\sigma^2_\theta+r(h_4-\dot\sigma^2_\theta))|_{r=0}|X\right]\\
=&\E \left[\frac{h_4(\theta_0)-\dot\sigma^2_{\theta_0}}{2\left(\sigma^2_{\theta_0}\right)^2}\big(\Lambda_{\theta_0}(Y)-m_{\theta_0}(X)\big)^2\bigg|X\right]\\&-\E \left[\dot \sigma_{\theta_0}^2\frac{h_2(\theta_0)-\sigma^2_{\theta_0}}{\left(\sigma^2_{\theta_0}\right)^3}\big(\Lambda_{\theta_0}(Y)-m_{\theta_0}(X)\big)^2\bigg|X\right]\\
&+\E \left[\frac{\dot \sigma_{\theta_0}^2}{\left(\sigma^2_{\theta_0}\right)^2}\big(\Lambda_{\theta_0}(Y)-m_{\theta_0}(X)\big)\big(m_{\theta_0}(X)-h_1(\theta_0,X)\big)\bigg|X\right]\\
=&\frac{h_4(\theta_0)-\dot\sigma^2_{\theta_0}}{2\left(\sigma^2_{\theta_0}\right)^2}\underbrace{\E \left[\big(\Lambda_{\theta_0}(Y)-m_{\theta_0}(X)\big)^2\big|X\right]}_{=\sigma_{\theta_0}^2}\\&-\dot \sigma_{\theta_0}^2\frac{h_2(\theta_0)-\sigma^2_{\theta_0}}{\left(\sigma^2_{\theta_0}\right)^3}\underbrace{\E \left[\big(\Lambda_{\theta_0}(Y)-m_{\theta_0}(X)\big)^2\big|X\right]}_{=\sigma_{\theta_0}^2}\\
&+\frac{\dot \sigma_{\theta_0}^2}{\left(\sigma^2_{\theta_0}\right)^2}\big(m_{\theta_0}(X)-h_1(\theta_0,X)\big)\underbrace{\E \left[\varepsilon_{\theta_0}\big|X\right]}_{=0}\\
=&\frac{h_4(\theta_0)-\dot\sigma^2_{\theta_0}}{2\sigma^2_{\theta_0}}-\dot \sigma_{\theta_0}^2\frac{h_2(\theta_0)-\sigma^2_{\theta_0}}{\left(\sigma^2_{\theta_0}\right)^2}.
\end{align*}
The conditions enable us to change derivation and integration, hence we obtain
\begin{align*}
&\quad D_0[h-h_0]\\&=\partial_r\Big\{\E\Big[\psi\Big((Y,X),\theta_0,h_0+r(h-h_0)\Big)\Big]\Big\}\Big|_{r=0}\\
&=\E\Big[\partial_r\psi\Big((Y,X),\theta_0,h_0+r(h-h_0)\Big)\Big|_{r=0}\Big]\\
&=\E\Big[\E\Big[\partial_r\psi\Big((Y,X),\theta_0,h_0+r(h-h_0)\Big)\Big|_{r=0}\Big|X\Big]\Big]\\
&=\E\bigg[-\E\bigg[\partial_rI(\theta_0,\sigma^2_\theta+r(h_2-\sigma^2_\theta),\dot\sigma^2_\theta+r(h_4-\dot\sigma^2_\theta))|_{r=0}|X\bigg]\\
&\quad-\E\bigg[\partial_rII(\theta_0,m_\theta +r(h_1-m_\theta),\sigma^2_\theta+r(h_2-\sigma^2_\theta),\dot m_\theta\\
&\quad+r(h_1-\dot m_\theta))|_{r=0}|X\bigg]+\E\bigg[\partial_rIII(\theta_0,m_\theta +r(h_1-m_\theta),\sigma^2_\theta+\\
&\quad r(h_2-\sigma^2_\theta),\dot\sigma^2_\theta+r(h_4-\dot\sigma^2_\theta))|_{r=0}|X\bigg]+\underbrace{\partial_r c_{\theta_0}|_{r=0}}_{=0}\bigg]\\
&=\E\bigg[-\frac{h_4(\theta_0)-\dot\sigma^2_{\theta_0}}{2\sigma^2_\theta}+\dot\sigma^2_\theta\frac{h_2(\theta_0)-\sigma^2_{\theta_0}}{2\left(\sigma^2_{\theta_0}\right)^2}+\frac{h_2(\theta_0)-\sigma^2_{\theta_0}}{\left(\sigma^2_{\theta_0}\right)^2}\E \left[\varepsilon_{\theta_0}\dot \varepsilon_{\theta_0}\big|X\right]\\
&\quad+\frac{h_4(\theta_0)-\dot\sigma^2_{\theta_0}}{2\sigma^2_{\theta_0}}-\dot \sigma_{\theta_0}^2\frac{h_2(\theta_0)-\sigma^2_{\theta_0}}{\left(\sigma^2_{\theta_0}\right)^2}\bigg]\\
&=0,
\end{align*}
where we used $\dot \sigma^2_{\theta_0}=2\E [\varepsilon_{\theta_0}\dot\varepsilon_{\theta_0}|X]$
in the last step.
\end{proof}
\newpage
\begin{proof}[Proof of Theorem \ref{Uniformlasso1}]\ \\
The assumptions \ref{A1}-\ref{A8} imply the conditions in theorem \ref{Uniformlasso} for the model \ref{model} and \ref{modeldot} choosing $\delta_n=O\left(n^{-1/4}\log(p\vee n)\right)$.\
\end{proof}
\begin{proof}[Proof of Theorem \ref{varianceest}]\ \\
As shown in the proof to Theorem \ref{Uniformlasso} the class of functions $$\mathcal{E}_{\Lambda}=\big\{\varepsilon_\theta=\left(\Lambda_{\theta}(\cdot)-\E\left[\Lambda_{\theta}(\cdot)|X\right]\right)|\theta\in\Theta\big\}$$
obeys
\begin{align*}
&\quad\log N(\varepsilon \|(F_{\Lambda}+F_{\Lambda}')\|_{\tilde{Q},2},\mathcal{E}_{\Lambda},L_2(\tilde{Q}))\\&\le \log N\left(\frac{\varepsilon}{2} \|F_{\Lambda}\|_{\tilde{Q},2},\mathcal{F}_{\Lambda},L_2(\tilde{Q})\right)+\log N\left(\frac{\varepsilon}{2} \|F'_{\Lambda}\|_{\tilde{Q},2},\mathcal{F}'_{\Lambda},L_2(\tilde{Q})\right)\\&\le \sup\limits_{Q}\log N\left(\frac{\varepsilon}{2} \|F_{\Lambda}\|_{Q,2},\mathcal{F}_{\Lambda},L_2(Q)\right)+\sup\limits_{Q'}\log N\left(\frac{\varepsilon}{2} \|F'_{\Lambda}\|_{Q',2},\mathcal{F}'_{\Lambda},L_2(Q')\right)\\&\le 2\sup\limits_{Q}\log N\left(\frac{\varepsilon^2}{16} \|F_{\Lambda}\|_{Q,2},\mathcal{F}_{\Lambda},L_2(Q)\right)\\&\le 4C_{\Lambda}'\log \left(4\sqrt{C_{\Lambda}''}/\varepsilon\right).
\end{align*}
Using Lemma 1 in the supplement to \cite{belloni2017program} we obtain with probabiliy $1-o(1)$
\begin{align*}
\sup\limits_{\theta\in\Theta}\frac{1}{n}\sum\limits_{i=1}^n|\varepsilon_{\theta,i}|\le \underbrace{\frac{1}{\sqrt{n}}\sup\limits_{g\in\mathcal{E}_{\Lambda}}G_n(g)}_{=O(\log(n)n^{-1/2})}+\underbrace{\sup\limits_{\theta\in\Theta}E[|\varepsilon_{\theta}|]}_{< C\ (\ref{B5})}=O(1).
\end{align*}
and
$$\sup\limits_{\theta\in\Theta}\frac{1}{n}\sum\limits_{i=1}^n\big(\varepsilon_{\theta,i}^2-\E [\varepsilon_\theta^2 ]\big)=\frac{1}{\sqrt{n}}\sup\limits_{g\in\mathcal{E}_{\Lambda}^2 }G_n(g)=O(\log(n)n^{-1/2})$$
as shown in the proof to Theorem \ref{Uniformlasso} with 
$$\mathcal{E}_{\Lambda}^2=\big\{\varepsilon_\theta^2=\left(\Lambda_{\theta}(\cdot)-\E\left[\Lambda_{\theta}(\cdot)|X\right]\right)^2|\theta\in\Theta\big\}.$$ 
We obtain with probabiliy $1-o(1)$
\begin{align*}
&\quad\sup\limits_{\theta\in\Theta}|\hat{\sigma}^2_{\theta}-\sigma^2_{\theta}|\\&=\sup\limits_{\theta\in\Theta}\bigg|\frac{1}{n}\sum\limits_{i=1}^n\Big(\underbrace{\Lambda_\theta(Y_i)-X_i^T\hat{\beta}_\theta}_{=\varepsilon_{\theta,i}-X_i^T\big(\hat{\beta}_\theta-\beta_\theta\big)}\Big)^2-\E [\varepsilon_\theta^2 ]\bigg|\\
&=\sup\limits_{\theta\in\Theta}\bigg|\frac{1}{n}\sum\limits_{i=1}^n\big(\varepsilon_{\theta,i}^2-\E [\varepsilon_\theta^2 ]\big)-\frac{2}{n}\sum\limits_{i=1}^n\varepsilon_{\theta,i}X_i^T\big(\hat{\beta}_\theta-\beta_\theta\big)\\&+\underbrace{\frac{1}{n}\sum\limits_{i=1}^n \Big(X_i^T\big(\hat{\beta}_\theta-\beta_\theta\big)\Big)^2}_{=||X^T(\hat{\beta}_\theta-\beta_\theta)||_{\mathbb{P}_{n,2}}^2}\bigg|\\
&\le \underbrace{\sup\limits_{\theta\in\Theta}\Big|\frac{1}{n}\sum\limits_{i=1}^n\big(\varepsilon_{\theta,i}^2-\E [\varepsilon_\theta^2 ]\big)\Big|}_{=O(\log(n)n^{-1/2})}+\underbrace{\sup\limits_{\theta\in\Theta}\Big|\frac{2}{n}\sum\limits_{i=1}^n\varepsilon_{\theta,i}X_i^T\big(\hat{\beta}_\theta-\beta_\theta\big)\Big|}_{\le 2K\sup\limits_{\theta\in\Theta}||\hat{\beta}_\theta-\beta_\theta||_1\frac{1}{n}\sum\limits_{i=1}^n|\varepsilon_{\theta,i}| }\\&+\underbrace{\sup\limits_{\theta\in\Theta}\Big| ||X^T(\hat{\beta}_\theta-\beta_\theta)||_{\mathbb{P}_{n,2}}^2\Big|}_{= O\Big(\frac{s\log(p\vee n)}{n}\Big)}\\
&\le 2K\underbrace{\sup\limits_{\theta\in\Theta}||\hat{\beta}_\theta-\beta_\theta||_1}_{= O\Big(\sqrt{\frac{s^2\log(p\vee n)}{n}}\Big)}\underbrace{\sup\limits_{\theta\in\Theta}\frac{1}{n}\sum\limits_{i=1}^n|\varepsilon_{\theta,i}|}_{= O(1)}+O\Big(\frac{s\log(p\vee n)}{n}\Big)+O(\log(n)n^{-1/2})\\
&=O\left(\max\left(\sqrt{\frac{s^2\log(p\vee n)}{n}},\frac{\log(n)}{n^{1/2}}\right)\right)=o\left(n^{-\frac{1}{4}}\right)
\end{align*}
with $K$ beeing a constant to bound the support of the regressors $X$.
As shown above and in the proof to Theorem \ref{Uniformlasso} the class of functions 
$$\dot{\mathcal{E}}_{\Lambda}=\big\{\dot\varepsilon_\theta=\left(\dot\Lambda_{\theta}(\cdot)-\E\left[\dot\Lambda_{\theta}(\cdot)|X\right]\right)|\theta\in\Theta\big\}$$
obeys
$$
\log N(\varepsilon \|(\dot F_{\Lambda}+\dot F_{\Lambda}')\|_{\tilde{Q},2},\dot{\mathcal{E}}_{\Lambda},L_2(\tilde{Q}))\le 4C_{\Lambda}'\log (4\sqrt{C_{\Lambda}''}/\varepsilon).
$$
As above using Lemma 1 in the supplement to \cite{belloni2017program} we obtain with probabiliy $1-o(1)$
$$\sup\limits_{\theta\in\Theta}\sum\limits_{i=1}^n|\dot{\varepsilon}_{\theta,i}|\le \underbrace{\sup\limits_{g\in\dot{\mathcal{E}}_{\Lambda}}G_n(g)}_{=O(\log(n)n^{-1/2})}+\underbrace{\sup\limits_{\theta\in\Theta}E[|\dot\varepsilon_{\theta}|]}_{< C}=O(1).$$
Since 
$$\E\left[\left(F_{\Lambda}+F_{\Lambda}'\right)^2\left(\dot F_{\Lambda}+\dot F_{\Lambda}'\right)^2\right]\le \left(\E\left[\left(F_{\Lambda}+F_{\Lambda}'\right)^4\right]\E\left[\left(\dot F_{\Lambda}+\dot F_{\Lambda}'\right)^4\right]\right)^{\frac{1}{2}}<\infty$$
an analogous argument to the one in the proof of \ref{Uniformlasso} applies and yields
$$\sup\limits_{\theta\in\Theta}\Big|\frac{1}{n}\sum\limits_{i=1}^n\big(\varepsilon_{\theta,i}\dot\varepsilon_{\theta,i} -\E [\varepsilon_\theta\dot\varepsilon_\theta  ]\big)\Big|=\frac{1}{\sqrt{n}}\sup\limits_{g\in\mathcal{E}_{\Lambda}\dot{\mathcal{E}_{\Lambda}} }\big|G_n(g)\big|=O(\log(n)n^{-1/2}).$$
We obtain with probabiliy $1-o(1)$
\begin{align*}
&\quad\sup\limits_{\theta\in\Theta}|\hat{\dot\sigma}^2_{\theta}-\dot\sigma^2_{\theta}|\\&=\sup\limits_{\theta\in\Theta}\bigg|\frac{2}{n}\sum\limits_{i=1}^n\Big(\underbrace{\Lambda_\theta(Y_i)-\hat{m}_\theta(X_i)}_{=\varepsilon_{\theta,i}-\big(\hat{m}_\theta(X_i)-m_\theta(X_i)\big)}\Big)\Big(\underbrace{\Lambda_\theta(Y_i)-\hat{\dot m}_\theta(X_i)}_{=\dot \varepsilon_{\theta,i}-\big(\hat{\dot m}_\theta(X_i)-\dot m_\theta(X_i)\big)}\Big)
\\&\quad-2\E [\varepsilon_\theta\dot\varepsilon_\theta ]\bigg|\\
&=2\sup\limits_{\theta\in\Theta}\bigg|\frac{1}{n}\sum\limits_{i=1}^n\big(\varepsilon_{\theta,i}\dot\varepsilon_{\theta,i} -\E [\varepsilon_\theta\dot\varepsilon_\theta  ]\big)-\frac{1}{n}\sum\limits_{i=1}^n\varepsilon_{\theta,i}\big(\hat{\dot m}_\theta(X_i)-\dot m_\theta(X_i)\big)\\
&\quad-\frac{1}{n}\sum\limits_{i=1}^n\dot\varepsilon_{\theta,i}\big(\hat{m}_\theta(X_i)-m_\theta(X_i)\big)\\&+\underbrace{\frac{1}{n}\sum\limits_{i=1}^n\big(\hat{m}_\theta(X_i)-m_\theta(X_i)\big)\big(\hat{\dot m}_\theta(X_i)-\dot m_\theta(X_i)\big)}_{\le\Big(\frac{1}{n}\sum\limits_{i=1}^n\big(\hat{m}_\theta(X_i)-m_\theta(X_i)\big)^2\Big)^{\frac{1}{2}}\Big(\frac{1}{n}\sum\limits_{i=1}^n\big(\hat{\dot m}_\theta(X_i)-\dot m_\theta(X_i)\big)^2\Big)^{\frac{1}{2}}}\bigg|\\
&\le 2\underbrace{\sup\limits_{\theta\in\Theta}\Big|\frac{1}{n}\sum\limits_{i=1}^n\big(\varepsilon_{\theta,i}\dot\varepsilon_{\theta,i} -\E [\varepsilon_\theta\dot\varepsilon_\theta  ]\big)\Big|}_{=O(\log(n)n^{-1/2})}+2\underbrace{\sup\limits_{\theta\in\Theta}\Big|\frac{1}{n}\sum\limits_{i=1}^n\varepsilon_{\theta,i}\big(\hat{\dot m}_\theta(X_i)-\dot m_\theta(X_i)\big)\Big|}_{\le 2K\sup\limits_{\theta\in\Theta}||\hat{\dot \beta}_\theta-\dot \beta_\theta||_1\frac{1}{n}\sum\limits_{i=1}^n|\varepsilon_{\theta,i}| }\\
&\quad +2\underbrace{\sup\limits_{\theta\in\Theta}\Big|\frac{1}{n}\sum\limits_{i=1}^n\dot\varepsilon_{\theta,i}\big(\hat{m}_\theta(X_i)-m_\theta(X_i)\big)\Big|}_{\le 2K\sup\limits_{\theta\in\Theta}||\hat{ \beta}_\theta- \beta_\theta||_1\frac{1}{n}\sum\limits_{i=1}^n|\dot\varepsilon_{\theta,i}| }\\
&\quad+\bigg(\underbrace{\sup\limits_{\theta\in\Theta}\Big| ||\big(\hat{m}_\theta(X_i)-m_\theta(X_i)\big)||_{\mathbb{P}_{n,2}}^2\Big|}_{= O\Big(\frac{s\log(p\vee n)}{n}\Big)}\underbrace{\sup\limits_{\theta\in\Theta}\Big| ||\big(\hat{\dot m}_\theta(X_i)-\dot m_\theta(X_i)\big)||_{\mathbb{P}_{n,2}}^2\Big|}_{= O\Big(\frac{s\log(p\vee n)}{n}\Big)}\Big)^{\frac{1}{2}}\\
&\le 2K\underbrace{\sup\limits_{\theta\in\Theta}||\hat{\dot \beta}_\theta-\dot \beta_\theta||_1}_{= O\Big(\sqrt{\frac{s^2\log(p\vee n)}{n}}\Big)}\underbrace{\sup\limits_{\theta\in\Theta}\frac{1}{n}\sum\limits_{i=1}^n|\varepsilon_{\theta,i}|}_{= O(1)}+2K\underbrace{\sup\limits_{\theta\in\Theta}||\hat{\beta}_\theta-\beta_\theta||_1}_{= O\Big(\sqrt{\frac{s^2\log(p\vee n)}{n}}\Big)}\underbrace{\sup\limits_{\theta\in\Theta}\frac{1}{n}\sum\limits_{i=1}^n|\dot\varepsilon_{\theta,i}|}_{= O(1)}\\
&\quad+O\Big(\frac{s\log(p\vee n)}{n}\Big)+O(\log(n)n^{-1/2})
\end{align*}
and therefore 
\begin{align*}
\quad\sup\limits_{\theta\in\Theta}|\hat{\dot\sigma}^2_{\theta}-\dot\sigma^2_{\theta}|=O\left(\max\left(\sqrt{\frac{s^2\log(p\vee n)}{n}},\frac{\log(n)}{n^{1/2}}\right)\right)=o\left(n^{-\frac{1}{4}}\right).
\end{align*}
\end{proof}
\newpage
\begin{proof}[Proof of Theorem \ref{entropy}]\ \\
The strategy of the proof is similar to the proof of theorem 1 from Belloni et al. (2014) \cite{belloni2014uniform}. Let $C,C_1$ and $C_2$ denote generic positive constants that may differ in each appearance, but do not depend on the sequence $P\in\mathcal{P}_n$.\\ 
For every $\theta\in\Theta$ the set $\tilde{\mathcal{H}}_1(\theta)$ consists of unions of $p$ choose $Cs$ sets, where the set of indices $\{i\in\{1,\dots,p\}:\beta_i\neq 0\}$ has cardinality not more than $Cs$ and therefore is a subset of a vector space with dimension $Cs$.
It follows that $\tilde{\mathcal{H}}_1(\theta)$ consists of unions of $p$ choose $Cs$ VC-subgraph classes $\tilde{\mathcal{H}}_{1,k}(\theta)$ with VC indices less or equal to $Cs+2$ (Lemma 2.6.15, Van der Vaart and Wellner (1996))\cite{vanweak}.\\
Using Theorem 2.6.7 in Van der Vaart and Wellner (1996) we obtain
\begin{align*}
&\quad\sup_{Q}\log N(\varepsilon\|\tilde{H}_1\|_{Q,2},\tilde{\mathcal{H}}_1(\theta),L_2(Q))\\&\le\sup_{Q}\log\Bigg(\sum\limits_{k=1}^{\binom{p}{Cs}}N(\varepsilon\|\tilde{H}_1\|_{Q,2},\tilde{\mathcal{H}}_{1,k}(\theta),L_2(Q))\Bigg)\\
&\le \sup_{Q}\log\Bigg( \underbrace{\binom{p}{Cs}}_{\le \big(\frac{e\cdot p}{Cs}\big)^{Cs}} K(Cs+2)(16e)^{Cs+2}\left(\frac{1}{\varepsilon}\right)^{2Cs+2}\Bigg)\\
&\le\log\Bigg( \left(\frac{e\cdot p}{Cs}\right)^{Cs} K(Cs+2)(16e)^{Cs+2}\left(\frac{1}{\varepsilon}\right)^{2Cs+2}\Bigg)\\
&\le Cs\log\Big(\frac{p}{\varepsilon}\Big)
\end{align*}
with $C$ beeing independent from $\theta$. Because
\begin{align*}
\sup\limits_{h_1(\theta)\in\tilde{\mathcal{H}}_1(\theta)}|h_1(\theta,x)|&\le\sup\limits_{\tilde{\beta}:\|\tilde{\beta}_\theta-\beta_\theta\|_1\le Cn^{-\frac{1}{4}}}|x\tilde{\beta}|\\
&\le\sup\limits_{\tilde{\beta}:\|\tilde{\beta}_\theta-\beta_\theta\|_1\le Cn^{-\frac{1}{4}}}|x\tilde{\beta}-x\beta_\theta|+|x\beta_\theta|\\
&\le KC+\E\left[F_\Lambda|X=x\right]=:\tilde{H}_1(x)
\end{align*}
the envelope $\tilde{H}_1$ can be choosen independent from $\theta$. Here and in the following we omit the dependence from $Y$ in $F_\Lambda\equiv F_\Lambda(Y)$ to simplify notation.\\
With the same argument we obtain
$$\sup_{Q}\log N(\varepsilon\|\tilde{H}_3\|_{Q,2},\tilde{\mathcal{H}}_3(\theta),L_2(Q)) \le Cs\log\Big(\frac{p}{\varepsilon}\Big)$$
with envelope $\tilde{H}_3(x):=KC+\E\left[\dot F_\Lambda|X=x\right]$.\\
Next we consider
\begin{align*}
\tilde{\mathcal{H}}_4(\theta):&=\left\{c\in\mathbb{R}\big|\ |c-\dot\sigma_\theta^2|\le Cn^{-1/4}\right\}=\left[\dot\sigma_\theta^2-Cn^{-1/4},\dot\sigma_\theta^2+Cn^{-1/4}\right]\\
&\subseteq\left[-(c+Cn^{-1/4}),(c+Cn^{-1/4})\right],
\end{align*}
where $c=\sup_{\theta\in\Theta}|\dot\sigma_\theta^2|<\infty$. This implies for all $\theta\in\Theta$
\begin{align*}
&\quad\sup_{Q}\log N(\varepsilon\|\tilde{H}_4\|_{Q,2},\tilde{\mathcal{H}}_4(\theta),L_2(Q))\\
&\le\sup_{Q}\log N\left(\varepsilon(c+C),\left[-(c+Cn^{-1/4}),c+Cn^{-1/4}\right],|\cdot|\right)\le \log\left(\frac{C}{\varepsilon}\right)
\end{align*}
with envelope $\tilde{H}_4=c+C$ and $C$ independet from $\theta$.\\ \\
Remark that $0<c_1=\inf_{\theta\in\Theta}\sigma_\theta^2$ and $c_2=\sup_{\theta\in\Theta}\sigma_\theta^2<\infty$ due to assumptions \ref{A5}-\ref{A6}. For $n$ sufficient large we find a $c_3$ with $0<c_3\le c_1-Cn^{-1/4}$. Therefore we can define
\begin{align*}
\bar{\mathcal{H}}_2(\theta):&=\left\{\frac{1}{\tilde{h}_2(\theta)}\big|\ \tilde{h}_2(\theta)\in\tilde{\mathcal{H}}_2(\theta)\right\}\\
&=\left\{1/c\big|\ |c-\sigma_\theta^2|\le Cn^{-1/4}\right\}\\
&=\left\{1/c\big|\ \frac{|c-\sigma_\theta^2|}{|c\sigma_\theta^2|}\le\frac{1}{|c\sigma_\theta^2|}Cn^{-1/4}\right\}\\
&\subseteq\left\{1/c\big|\ \frac{|c-\sigma_\theta^2|}{|c\sigma_\theta^2|}\le C^*n^{-1/4}\right\}\\
&=\left\{\bar{c}\big|\ |\bar{c}-1/\sigma_\theta^2|\le C^*n^{-1/4}\right\}\\
&=\left[1/\sigma_\theta^2-C^*n^{-1/4},1/\sigma_\theta^2+C^*n^{-1/4}\right]\\
&\subseteq\left[1/c_2-C^*n^{-1/4},1/c_1+C^*n^{-1/4}\right]
\end{align*}
with $C^*=\frac{C}{c_3c_1}$. Analogous as above we obtain for all $\theta\in\Theta$
\begin{align*}
\sup_{Q}\log N(\varepsilon\|\bar{H}_2\|_{Q,2},\bar{\mathcal{H}}_2(\theta),L_2(Q))\le \log\left(\frac{C}{\varepsilon}\right)
\end{align*}
with envelope $\bar{H}_2=1/c_2+C^*$ and $C$ independet from $\theta$. Define
\begin{align*}
I(\theta,\bar{\mathcal{H}}_2,\tilde{\mathcal{H}}_4)&:=\left\{-\frac{1}{2}h_4(\theta)h_2(\theta)|\ h_4(\theta)\in\tilde{\mathcal{H}}_4(\theta),h_2(\theta)\in\bar{\mathcal{H}}_2(\theta)\right\},
\end{align*}
\begin{align*}
&\quad II(\theta,\tilde{\mathcal{H}}_1,\bar{\mathcal{H}}_2,\tilde{\mathcal{H}}_3)\\&:=\Bigg\{(y,x)\mapsto-h_2(\theta)\left(\Lambda_\theta(y)-h_1(\theta,x)\right)\left(\dot\Lambda_\theta(y)-h_3(\theta,x)\right)\\&|\ h_1(\theta)\in\tilde{\mathcal{H}}_1(\theta),h_2(\theta)\in\bar{\mathcal{H}}_2(\theta),h_3(\theta)\in\tilde{\mathcal{H}}_3(\theta)\Bigg\}
\end{align*}
and
\begin{align*}
III(\theta,\tilde{\mathcal{H}}_1,\bar{\mathcal{H}}_2,\tilde{\mathcal{H}}_4)&:=\Bigg\{(y,x)\mapsto \frac{1}{2}h_2^2(\theta)h_4(\theta)\left(\Lambda_\theta(y)-h_1(\theta,x)\right)^2\\
&\quad\quad\quad|\ h_1(\theta)\in\tilde{\mathcal{H}}_1(\theta),h_2(\theta)\in\bar{\mathcal{H}}_2(\theta),h_4(\theta)\in\tilde{\mathcal{H}}_4(\theta)\Bigg\}.
\end{align*}
By Lemma L.1 in the supplement to \cite{belloni2017program} we have
\begin{align*}
&\quad\log N\left(\varepsilon\|1/2\bar{H}_2\tilde{H}_4\|_{Q,2},I(\theta,\bar{\mathcal{H}}_2,\tilde{\mathcal{H}}_4),L_2(Q)\right)\\
&\le\log N\left(\frac{\varepsilon}{4}\|\bar{H}_2\|_{Q,2},\bar{\mathcal{H}}_2(\theta),L_2(Q)\right)+\log N\left(\frac{\varepsilon}{4}\|\tilde{H}_4\|_{Q,2},\tilde{\mathcal{H}}_4(\theta),L_2(Q)\right)\\
&\le 2\log\left(\frac{C}{\varepsilon}\right).
\end{align*}
With \ref{A6} we obtain
\begin{align*}
&\quad\log N\bigg(\varepsilon\|\bar{H}_2(F_\Lambda+\tilde{H}_1)(\dot F_\Lambda+\tilde{H}_3)\|_{Q,2},II(\theta,\tilde{\mathcal{H}}_1,\bar{\mathcal{H}}_2,\tilde{\mathcal{H}}_3),L_2(Q)\bigg)\\
&\le\log N\left(\frac{\varepsilon}{2}\|\bar{H}_2(\theta)\|_{Q,2},\bar{\mathcal{H}}_2(\theta),L_2(Q)\right)\\
&\quad+\log N\left(\frac{\varepsilon}{4}\|(F_\Lambda+\tilde{H}_1)\|_{Q,2},\mathcal{F}_\Lambda-\tilde{\mathcal{H}}_1(\theta),L_2(Q)\right)\\
&\quad+\log N\left(\frac{\varepsilon}{4}\|(\dot{F}_\Lambda+\tilde{H}_3)\|_{Q,2},\dot{\mathcal{F}}_\Lambda-\tilde{\mathcal{H}}_3(\theta),L_2(Q)\right)\\
&\le \log\left(\frac{2C}{\varepsilon}\right)+\log N\left(\frac{\varepsilon}{8}\|F_\Lambda\|_{Q,2},\mathcal{F}_\Lambda,L_2(Q)\right)\\
&\quad+\log N\left(\frac{\varepsilon}{8}\|\dot{F}_\Lambda\|_{Q,2},\dot{\mathcal{F}}_\Lambda,L_2(Q)\right)\\
&\quad+\log N\left(\frac{\varepsilon}{8}\|\tilde{H}_1\|_{Q,2},\tilde{\mathcal{H}}_1(\theta),L_2(Q)\right)\\
&\quad+\log N\left(\frac{\varepsilon}{8}\|\tilde{H}_3\|_{Q,2},\tilde{\mathcal{H}}_3(\theta),L_2(Q)\right)\\
&\le \log\left(\frac{2C}{\varepsilon}\right)+ C_{\Lambda}'\log (8C_{\Lambda}''/\varepsilon)+\dot{C}_{\Lambda}'\log (8\dot{C}_{\Lambda}''/\varepsilon)+Cs\log\Big(\frac{8p}{\varepsilon}\Big)\\
&\quad+Cs\log\Big(\frac{8p}{\varepsilon}\Big)\\
&\le C_1s\log\Big(\frac{C_2 p}{\varepsilon}\Big)
\end{align*}
and with an analogous argument
\begin{align*}
&\quad\log N\left(\varepsilon\|\frac{1}{2}\bar{H}^2_2\tilde{H}_4(F_\Lambda+\tilde{H}_1)^2\|_{Q,2},III(\theta,\tilde{\mathcal{H}}_1,\bar{\mathcal{H}}_2,\tilde{\mathcal{H}}_4),L_2(Q)\right)\\&\le C_1s\log\Big(\frac{C_2 p}{\varepsilon}\Big).
\end{align*}
Because
$$\Psi(\theta)=I(\theta,\bar{\mathcal{H}}_2,\tilde{\mathcal{H}}_4)+II(\theta,\tilde{\mathcal{H}}_1,\bar{\mathcal{H}}_2,\tilde{\mathcal{H}}_3)+III(\theta,\tilde{\mathcal{H}}_1,\bar{\mathcal{H}}_2,\tilde{\mathcal{H}}_4)+c_\theta$$
we can define the envelope
\begin{align*}
\bar\psi(Y,X)&:=\frac{1}{2}\bar{H}_2\tilde{H}_4+\bar{H}_2(F_\Lambda+\tilde{H}_1)(\dot F_\Lambda+\tilde{H}_3)\\
&\quad+\frac{1}{2}\bar{H}^2_2\tilde{H}_4(F_\Lambda+\tilde{H}_1)^2+J_\Lambda,
\end{align*}
which is independent from $\theta$ with
\begin{align*}
&\quad\E\left[\left(\bar\psi(Y,X)\right)^4\right]\\&=\E\left[\left(\frac{1}{2}\bar{H}_2\tilde{H}_4+\bar{H}_2(F_\Lambda+\tilde{H}_1)(\dot F_\Lambda+\tilde{H}_3)+\frac{1}{2}\bar{H}^2_2\tilde{H}_4(F_\Lambda+\tilde{H}_1)^2+J_\Lambda\right)^4\right]\\
&<\infty
\end{align*}
where we used \ref{A6} and \ref{A9}. Additionally by using $N(\varepsilon||J_\Lambda||_{Q,2},c_\theta,L_2(Q))=1$ for all $\theta\in\Theta$ and Lemma L.1 in the supplement to \cite{belloni2017program} we obtain 
$$\sup_{Q}\log N(\varepsilon||\bar\psi||_{Q,2},\Psi(\theta),L_2(Q))\le C_1 s\log\left(\frac{C_2(p\vee n)}{\varepsilon}\right),$$
where the supremum is taken over all probability measures $Q$ with\\ $\E_Q\left[\left(\bar\psi(Y,X)\right)^2\right]<\infty$.
\end{proof}
\newpage

\begin{proof}[Proof of Theorem \ref{main}]\ \\
We demonstrate that conditions \ref{C1}-\ref{C7} from theorem \ref{generalmain} are satisfied. Most conditions are already proven in the preceding theorems.
The condition \ref{C1} is shown in lemma \ref{identification}. 
Due to theorem \ref{Uniformlasso1} and \ref{varianceest} condition \ref{C4} is satisfied with $\tilde{\mathcal{H}}$ and $\tilde{\mathcal{H}}(\theta)$ as defined in \ref{entropysec}. Condition \ref{C5} is proved in theorem \ref{entropy}.
Again, choosing $\mathcal{H}'=\tilde{\mathcal{H}}$ as defined in \ref{entropysec} the conditions in lemma \ref{orthogonalityprop} hold
where we used \ref{boundsigma} and the envelope in \ref{C5} which implies \ref{C2}.
Since conditions \ref{C3} and \ref{C7} are the same as \ref{A12} and \ref{A11}, we need to verify \ref{C6}. Due to condition \ref{A3}, choosing $\rho_n=o(n^{-1/4})$, we have
\begin{align*}
&\quad\sup\limits_{\theta\in\Theta,\tilde{h}\in\tilde{\mathcal{H}}(\theta)}|\E[\psi((Y,X)),\theta,h_0(\theta)]-\E[\psi((Y,X)),\theta,\tilde{h}(\theta)]|\\
&\le \sup\limits_{\theta\in\Theta,\tilde{h}\in\tilde{\mathcal{H}}(\theta)}\E\Big[\Big(\psi\big((Y,X),\theta,\tilde{h}(\theta,X)\big)-\psi\big((Y,X),\theta,h_0(\theta,X)\big)\Big)^2\Big]^{\frac{1}{2}}\\
&\le \sup\limits_{\theta\in\Theta,\tilde{h}\in\tilde{\mathcal{H}}(\theta)}C\E\left[\|\tilde{h}(\theta,X)-h_0(\theta,X)\|^{2}_2\right]^{\frac{1}{2}}\\
&\le C\rho_n,
\end{align*}
where we used \ref{A10} (ii) and
\begin{align*}
E\Big[\|\tilde{h}(\theta,X)-h_0(\theta,X)\|^{2}_2\Big]&=\E\big[(\tilde{h}_1(\theta,X)-m_\theta(X))^{2}\big]+\E\big[(\tilde{h}_2(\theta)-\sigma^2_\theta)^{2}\big]\\
&\quad+\E\big[(\tilde{h}_3(\theta,X)-\dot m_\theta(X))^{2}\big]+\E\big[(\tilde{h}_4(\theta)-\dot\sigma^2_\theta)^{2}\big]\\
&\le \rho_n^2.
\end{align*}
The last inequality follows from the properties of $\tilde{\mathcal{H}}$. We have
\begin{align*}
\E\big[(\tilde{h}_1(\theta,X)-m_\theta(X))^2\big]&=\E\left[\left(X(\tilde{\beta}_\theta-\beta_\theta)\right)^2\right]\\
&\le K^2\sup\limits_{\theta\in\Theta}||\tilde{\beta}_\theta-\beta_\theta||^2_{1}\\
&\le \rho_n^2,
\end{align*}
\begin{align*}
\E\big[(\tilde{h}_2(\theta)-\sigma^2_\theta)^{2}\big] 
&\le \rho_n^2
\end{align*}
and the same holds for the two remaining terms with an analogous argument. Therefore \ref{C6} (i) holds.\\ \\
In the following, we take the supremum over all $\theta$ with $|\theta-\theta_0|\le C\rho_n$ and $\tilde{h}\in\tilde{\mathcal{H}}(\theta)$, meaning
$$\sup\equiv\sup\limits_{\theta:|\theta-\theta_0|\le C\rho_n,\tilde{h}\in\tilde{\mathcal{H}}(\theta)}.$$
By \ref{A10} (i) and (ii), we have
\begin{align*}
&\quad\sup\E\Big[\Big(\psi\big((Y,X),\theta,\tilde{h}(\theta,X)\big)-\psi\big((Y,X),\theta_0,h_0(\theta_0,X)\big)\Big)^2\Big]^{1/2}\\
&= \sup\E\bigg[\Big(\psi\big((Y,X),\theta,\tilde{h}(\theta,X)\big)-\psi\big((Y,X),\theta,h_0(\theta,X)\big)\\
&\quad\quad+\psi\big((Y,X),\theta,h_0(\theta,X)\big)-\psi\big((Y,X),\theta_0,h_0(\theta_0,X)\big)\Big)^2\bigg]^{1/2}\\
&\le \sup\E\bigg[\Big(\psi\big((Y,X),\theta,\tilde{h}(\theta,X)\big)-\psi\big((Y,X),\theta,h_0(\theta,X)\big)\Big)^2\\
&\quad+\Big(\psi\big((Y,X),\theta,h_0(\theta,X)\big)-\psi\big((Y,X),\theta_0,h_0(\theta_0,X)\big)\Big)^2\\
&\quad+2\Big(\psi\big((Y,X),\theta,\tilde{h}(\theta,X)\big)-\psi\big((Y,X),\theta,h_0(\theta,x)\big)\Big)\\
&\quad\quad \Big(\psi\big((Y,X),\theta,h_0(\theta,X)\big)-\psi\big((Y,X),\theta_0,h_0(\theta_0,X)\big)\Big)\bigg]^{1/2}\\
&\le\sup C\bigg(E\Big[\|\tilde{h}(\theta,X)-h_0(\theta,X)\|^{2}_2\Big]+|\theta-\theta_0|^{2}\\
&\quad+|\theta-\theta_0|\sqrt{E\Big[\|\tilde{h}(\theta,X)-h_0(\theta,X)\|^{2}_2\Big]}\bigg)^{1/2}\\
&\le C\rho_n\\
&\le Cn^{-1/4}.
\end{align*}
Because of the growth condition \ref{A3} we have
\begin{align*}
n^{-1/4} s^{\frac{1}{2}}\log\Big(\frac{(p\vee n) }{n^{-1/4}}\Big)^{\frac{1}{2}}+n^{-\frac{1}{2}+\frac{1}{q}}s\log\Big(\frac{(p\vee n) }{n^{-1/4}}\Big)=o(1)
\end{align*}
and \ref{C6} (ii) follows.\\ \\
Condition $\ref{C6}$ (iii) follows directly from $\ref{A10}$ (iii): 
\begin{align*}
&\quad\sup\limits_{r\in(0,1)}\sup\left|\partial^2_r\bigg\{\E\Big[\psi\big((Y,X),\theta_0+r(\theta-\theta_0),h_0+r(\tilde{h}-h_0)\big)\Big]\bigg\}\right|\\
&\le \sup\limits_{\theta:|\theta-\theta_0|\le C\rho_n,\tilde{h}\in\tilde{\mathcal{H}}(\theta)}C\left(|\theta-\theta_0|^2+\sup\limits_{\theta^*\in\Theta}\E\Big[\|\tilde{h}(\theta^*,X)-h_0(\theta^*,X)\|^{2}_2\Big]\right)\\
&\le C\rho_n^2\\
&=o(n^{-1/2}).
\end{align*}
\end{proof}
\begin{proof}[Proof of lemma \ref{VC}]
\begin{remark}
The proof for Box-Cox-Transformations is from \cite{vanweak}, who refer to \cite{quiroz1996estimation}. It heavily relies on the properties of the dual density from \cite{assouad1983densite}. We give a detailed version of the proof of \cite{quiroz1996estimation} and extend the idea to the class of derivatives and Yeo-Johnson Power Transformations.
\end{remark}\ \\
Since adding a single function to a class of functions can increase the VC index at most by one, we exclude the parameter $\theta=0$ from the proof and restrict the class to
$$\mathcal{F}'_1=\big\{\Lambda_{\theta}(\cdot)|\theta\in\mathbb{R}\setminus\{0\}\big\}.$$
At first recall that $\mathcal{F}'_1$ is a VC class if and only if the between graph set 
$$\mathcal{C}:=\big\{C_\theta|\theta\in\mathbb{R}\setminus\{0\}\big\}$$
with
$$C_\theta:=\big\{(x,t)\in \mathbb{R}^+\times\mathbb{R}|0\le t\le \Lambda_\theta(x)\text{ or } \Lambda_\theta(x)\le t\le 0\big\}$$
is a VC class (cf. \cite{vanweak}, page 152).
We now consider the dual class (cf. \cite{assouad1983densite}) of $\mathcal{C}$ given by
$$\mathcal{D}:=\big\{D_{(x,t)}|(x,t)\in \mathbb{R}^+\times\mathbb{R}\big\}$$
with
\begin{align*}
D_{(x,t)}:&=\big\{\theta\in \mathbb{R}\setminus\{0\}|(x,t)\in C_\theta\big\}\\
&=\big\{\theta\in \mathbb{R}\setminus\{0\}|0\le t\le \Lambda_\theta(x)\text{ or } \Lambda_\theta(x)\le t\le 0\big\}.
\end{align*}
For the derivative of $\Lambda_\theta(x)$ we have
\begin{alignat*}{2}
&&\dot\Lambda_\theta(x)=\frac{1}{\theta^2}\left(\left(\theta \log(x)-1\right)x^\theta+1\right)&\ge 0\\
\Leftrightarrow&&\left(\theta \log(x)-1\right)x^\theta&\ge-1\\
\Leftrightarrow&& \log(x^\theta)&\ge\frac{x^\theta-1}{x^\theta},
\end{alignat*}
which is true for all $x$ and $\theta$.
Since $\Lambda_\theta(x)$ is continuous and monotone increasing in $\theta$ the set $D_{(x,t)}$ is the union of at most two intervals in $\mathbb{R}\setminus\{0\}$ and therefore $\mathcal{D}$ is a VC class, which by proposition 2.12 in \cite{assouad1983densite} implies that $\mathcal{C}$ is a VC class.\\
With the same argument as above we have to prove that
$$\mathcal{D}'=\big\{D'_{(x,t)}|(x,t)\in \mathbb{R}^+\times\mathbb{R}\big\}$$
is a VC class with
$$D'_{(x,t)}:=\big\{\theta\in \mathbb{R}\setminus\{0\}|0\le t\le \dot\Lambda_\theta(x)\big\}$$
since $\dot\Lambda_\theta(x)\ge 0$. The second derivative with respect to $\theta$ is given by
$$\ddot \Lambda_\theta(x)=\frac{1}{\theta^3}\Big(\underbrace{\left(\log(x^\theta)-1\right)^2x^\theta+x^\theta-2}_{=:f(x^\theta)}\Big).$$
The case $x=1$ directly implies $\ddot \Lambda_\theta(x)=0$. Substitue $z=x^\theta$ in $f(x^\theta)$ and see that
$$f'(z)=\left(\log(z)-1\right)^2+2\left(\log(z)-1\right)+1=\left(\log(z)\right)^2\ge 0.$$
This together with $f(1)=0$ implies $f(z)\ge 0$ for $z\ge 1$ and  $f(z)< 0$ for $z< 1$. The four cases
  $$\begin{array}{cl}
   \left.\begin{aligned}
     &x>1,\  \theta>0 \\
     0<&x<1,\  \theta<0 
   \end{aligned}\right\} &\Rightarrow x^\theta>1\\
   \left.\begin{aligned}
     &x>1,\ \theta<0 \\
     0<&x<1,\ \theta>0 
   \end{aligned}\right\} &\Rightarrow 0<x^\theta<1
 \end{array}$$
and the coefficient $1/\theta^3$ imply
$$\ddot \Lambda_\theta(x)=\begin{cases} \ge 0 \text{ for } x\ge 1\\ < 0 \text{ for } x\le 1.\end{cases}$$
We have that $\dot\Lambda_\theta(x)$ is continuous in $\theta$, monotone increasing for $x\ge 1$ and monotone decreasing for $x<1$. This again implies that the set $D_{(x,t)}$ is the union of at most two intervals in $\mathbb{R}\setminus\{0\}$.
We now consider the class of Yeo-Johnson Power Transformations
$$\mathcal{F}_2=\big\{\Psi_{\theta}(\cdot)|\theta\in\mathbb{R}\setminus\{0,2\}\big\},$$ 
where we exclude the parameters $\theta=0$ and $\theta=2$. The between graph set is given by
$$\tilde{\mathcal{C}}:=\big\{\tilde{C}_\theta|\theta\in\mathbb{R}\setminus\{0,2\}\big\}$$
with 
$$\tilde{C}_\theta:=\big\{(x,t)\in \mathbb{R}\times\mathbb{R}|0\le t\le \Psi_\theta(x)\text{ or } \Psi_\theta(x)\le t\le 0\big\}.$$
Since $\Psi_\theta(x)\ge 0$ for $x\ge 0$ and $\Psi_\theta(x)< 0$ for $x<0$ we have
\begin{align*}
\tilde{C}_\theta:&=\big\{(x,t)\in \mathbb{R}\times\mathbb{R}|0\le t\le \Psi_\theta(x)\text{ or } \Psi_\theta(x)\le t\le 0\big\}\\
&=\big\{(x,t)\in \mathbb{R}^+_0\times\mathbb{R}|0\le t\le \Psi_\theta(x)\big\}\cup\big\{(x,t)\in \mathbb{R}^-\times\mathbb{R}|\Psi_\theta(x)\le t\le 0\big\}\\
&=\underbrace{\big\{(x,t)\in \mathbb{R}^+_0\times\mathbb{R}|0\le t\le \Lambda_\theta(x+1)\big\}}_{=:\tilde{C}_{\theta,1}}\\&\quad\cup\underbrace{\big\{(x,t)\in \mathbb{R}^-\times\mathbb{R}|-\Lambda_{2-\theta}(-x+1)\le t\le 0\big\}}_{=:\tilde{C}_{\theta,2}}.
\end{align*}
The sets
$$\tilde{\mathcal{C}_1}:=\big\{\tilde{C}_{\theta,1}|\theta\in\mathbb{R}\setminus\{0,2\}\big\}\text{ and }\tilde{\mathcal{C}}_2:=\big\{\tilde{C}_{\theta,2}|\theta\in\mathbb{R}\setminus\{0,2\}\big\}$$
are VC-classes as shown above. Using Lemma 2.6.17 (iii) from \cite{vanweak} we obtain that
$$\tilde{\mathcal{C}}_1\sqcup \tilde{\mathcal{C}}_2=\big\{\tilde{C}_{\theta,1}\cup \tilde{C}_{\theta,2}|\tilde{C}_{\theta,1}\in\tilde{\mathcal{C}}_1,\tilde{C}_{\theta,2}\in\tilde{\mathcal{C}}_2\big\}$$ is a VC-class which contains $\tilde{\mathcal{C}}$. The proof for the class of derivatives can be shown analogously.
\end{proof}
\newpage

\section*{Supplementary Material}
\section{Uniform convergence rates for the lasso}\label{appendixB}
Assumptions \textbf{B1}-\textbf{B7}.\\
The following assumptions hold uniformly in $n\ge n_0,P\in\mathcal{P}_n$:
\begin{enumerate}[label=\textbf{B\arabic*},ref=B\arabic*]
\item\label{B1}
Uniformly in $\theta$ the model is sparse, namely $\sup\limits_{\theta\in\Theta}\|\beta_\theta\|_0\le s$
\item \label{B2}
The parameters obey the growth conditions  $n^{-1/4}\log(p\vee n)\le \delta_n$ and $s\log(p\vee n)\le \delta_nn$ for $\delta_n\searrow 0$ approaching zero from above at a speed at most polynomial in $n$.
\item \label{B3}
For all $n\in\mathbb{N}$ the regressor $X=(X_1,\dots,X_p)$ has a bounded support $\mathcal{X}$.
\item \label{B4}
Uniformly in $\theta$ the variance of the error term is bounded from zero 
    $$0<c\le \inf\limits_{\theta\in\Theta}\E\big[\varepsilon_\theta^2\big].$$
\item \label{B5}
The transformations are measurable and the class of transformations 
    $$\mathcal{F}_{\Lambda}:=\big\{\Lambda_{\theta}(\cdot)|\theta\in\Theta\big\}$$
    has VC index $C_{\Lambda}$ and an envelope $F_{\Lambda}$ with 
    $$\E [F_{\Lambda}(Y)^6]<\infty.$$ 
\item \label{B6}
The transformations are differentiable with respect to $\theta$ and the following condition holds:
    $$\sup\limits_{\theta\in\Theta}\E \Big[\big(\dot\Lambda_\theta(Y)\big)^2\Big]\le C.$$
\item \label{B7}
 With probability $1-o(1)$ the empirical minimum and maximum sparse eigenvalues are bounded from zero and above, namely
    \begin{align*}
		0<\kappa' &\le\inf\limits_{||\delta||_0\le s\log(n),||\delta||=1}||X^T\delta||_{\mathbb{P}_{n,2}}\\
		&\le\sup\limits_{||\delta||_0\le s\log(n),||\delta||=1}||X^T\delta||_{\mathbb{P}_{n,2}}\le\kappa''<\infty.
		\end{align*}
\end{enumerate}\ \\
\begin{theorem}\label{Uniformlasso}\ \\
Under assumptions \ref{B1}-\ref{B7} above, uniformly for all $P\in\mathcal{P}_n$ with prohability $1-o(1)$, it holds:
\begin{enumerate}
\item $\sup\limits_{\theta\in\Theta}||\hat{\beta}_\theta||_0=O(s)$.
\item $\sup\limits_{\theta\in\Theta}||X^T(\hat{\beta}_\theta-\beta_\theta)||_{\mathbb{P}_{n,2}}=O\left(\sqrt{\frac{s\log(p\vee n)}{n}}\right)$.
\item $\sup\limits_{\theta\in\Theta}||\hat{\beta}_\theta-\beta_\theta||_{1}=O\left(\sqrt{\frac{s^2\log(p\vee n)}{n}}\right)$.
\end{enumerate}
\end{theorem}\ \\

\begin{proof}
We verify the Assumption 6.1 from Belloni et al. (2017) \cite{belloni2017program}. Due to assumptions $\ref{B1}$ and $\ref{B2}$ the condition 6.1(i) is satisfied. Needless to say, the assumption 6.1(ii) holds for a compact $\Theta\subset\mathbb{R}$. Remark that assumption \ref{B5} implies the conditions 6.1 (iii)
which follows from \ref{boundsigma}. Due to assumption \ref{B3} the conditions in 6.1(iv)(a) are satisfied and we can omit the $X$ in the technical conditions in 6.1(iv)(b). The eigenvalue condition 6.1(iv)(c) is the same as in \ref{B7}. Therefore we have to show with probability $1-o(1)$:
\begin{itemize}
\item[(1)] $\sup\limits_{\theta\in\Theta}|(\mathbb{E}_n-\E )\varepsilon^2_\theta\vee(\mathbb{E}_n-\E )\Lambda_\theta(Y)^2|=O\left(\delta_n\right)$
\item[(2)] $n^{1/2}\sup\limits_{|\theta-\theta'|\le 1/n}|\mathbb{E}_n\left[
\varepsilon_\theta-\varepsilon_{\theta'}\right]|=O\left(\delta_n\right)$ and
\item[(3)] $\log(p\vee n)^{1/2}\sup\limits_{|\theta-\theta'|\le 1/n}\mathbb{E}_n\left[(
\varepsilon_\theta-\varepsilon_{\theta'})^2\right]^{1/2}=O\left(\delta_n\right).$
\end{itemize}
Because $\mathcal{F}_{\Lambda}$ is a VC-class of functions with VC index $C_{\Lambda}$, we have by Theorem 2.6.7 in \cite{vanweak}
\begin{align}\label{B.1}
\log N(\varepsilon \|F_{\Lambda}\|_{Q,2},\mathcal{F}_{\Lambda},L_2(Q))\le C_{\Lambda}'\log (C_{\Lambda}''/\varepsilon),
\end{align}
for any $Q$ with $ \|F_{\Lambda}\|_{Q,2}^2=\mathbb{E}_{Q}[F^2_{\Lambda}]<\infty$, where the constants $C_{\Lambda}'$ and $C_{\Lambda}''$ only depend on the VC index. Define
$$\mathcal{F}'_{\Lambda}:=\big\{\E\left[\Lambda_{\theta}(\cdot)|X\right]|\theta\in\Theta\big\}$$ with envelope $F'_{\Lambda}:=E[F_{\Lambda}|X]$ and $$\mathcal{E}_{\Lambda}^2:=\big\{\left(\Lambda_{\theta}(\cdot)-\E\left[\Lambda_{\theta}(\cdot)|X\right]\right)^2|\theta\in\Theta\big\}.$$
with envelope $\left(F_{\Lambda}+F'_\Lambda\right)^2$. By Lemma L.2 in the supplement to \cite{belloni2017program} we have
\begin{align}\label{B.2}
\sup\limits_{Q'}\log N(\varepsilon \|F'_{\Lambda}\|_{Q',2},\mathcal{F}'_{\Lambda},L_2(Q'))\le \sup\limits_{Q}\log N((\varepsilon/4)^2 \|F_{\Lambda}\|_{Q,2},\mathcal{F}_{\Lambda},L_2(Q)),
\end{align}
where the supremum on the left-hand side is taken over all probability measures $Q'$ with $$\|F'_{\Lambda}\|_{Q',2}^2:=\mathbb{E}_{Q'}\Big[\big(\mathbb{E}[F_{\Lambda}(Y)|X]\big)^2\Big]\equiv\mathbb{E}_{Q'}\Big[\big(\mathbb{E}[F_{\Lambda}|X]\big)^2\Big]<\infty.$$
Since $\mathcal{E}_{\Lambda}^2\subset (\mathcal{F}_{\Lambda}-\mathcal{F}'_{\Lambda})^2$ it follows by Lemma L.1 in the supplement to \cite{belloni2017program} for any $\tilde{Q}$ with $\mathbb{E}_{\tilde{Q}}[(F_{\Lambda}+F'_{\Lambda})^4]<\infty$ and $0<\varepsilon\le 1$
\begin{align*}
&\quad\log N(\varepsilon \|(F_{\Lambda}+F'_{\Lambda})^2\|_{\tilde{Q},2},\mathcal{E}_{\Lambda}^2,L_2(\tilde{Q}))\\&\le
2\log N\left(\frac{\varepsilon}{2}\|F_{\Lambda}+F'_{\Lambda}\|_{\tilde{Q},2},\mathcal{F}_{\Lambda}-\mathcal{F}'_{\Lambda},L_2(\tilde{Q})\right)\\&\le 2\log N\left(\frac{\varepsilon}{4} \|F_{\Lambda}\|_{\tilde{Q},2},\mathcal{F}_{\Lambda},L_2(\tilde{Q})\right)+2\log N\left(\frac{\varepsilon}{4} \|F'_{\Lambda}\|_{\tilde{Q},2},\mathcal{F}'_{\Lambda},L_2(\tilde{Q})\right)\\&\le 2\sup\limits_{Q}\log N\left(\frac{\varepsilon}{4} \|F_{\Lambda}\|_{Q,2},\mathcal{F}_{\Lambda},L_2(Q)\right)\\&\quad+2\sup\limits_{Q'}\log N\left(\frac{\varepsilon}{4} \|F'_{\Lambda}\|_{Q',2},\mathcal{F}'_{\Lambda},L_2(Q')\right)\\&\le 4\sup\limits_{Q}\log N\left(\frac{\varepsilon^2}{256} \|F_{\Lambda}\|_{Q,2},\mathcal{F}_{\Lambda},L_2(Q)\right),
\end{align*}
where we used \ref{B.2} in the last step.  We conclude
\begin{align*}
\log N(\varepsilon \|(F_{\Lambda}+F'_{\Lambda})^2\|_{\tilde{Q},2},\mathcal{E}_{\Lambda}^2,L_2(\tilde{Q}))&\le 4C_{\Lambda}'\log (256C_{\Lambda}''/\varepsilon^2)\\&=16C_{\Lambda}'\log (16\sqrt{C_{\Lambda}''}/\varepsilon)
\end{align*}
by \ref{B.1}. Under \ref{B5} for all $r \in\{1,2,3\}$ it holds
$$\E\left[F'^{2r}_{\Lambda}\right]=\E\left[\left(\E\left[F_{\Lambda}|X\right]\right)^{2r}\right]\le \E\left[\E\left[\left(F_{\Lambda}\right)^{2r}|X\right]\right]=\E\left[F^{2r}_{\Lambda}\right]<\infty,$$
which implies
\begin{align*}
\mathbb{E}\left[(F_{\Lambda}+F'_{\Lambda})^4\right]&=\mathbb{E}\left[F^4_{\Lambda}\right]+\underbrace{\mathbb{E}\left[F'^4_{\Lambda}\right]}_{\le\mathbb{E}\left[F^4_{\Lambda}\right]}+6\underbrace{\mathbb{E}\left[F^2_{\Lambda}F'^2_{\Lambda}\right]}_{\le \sqrt{\mathbb{E}\left[F^4_{\Lambda}\right]\mathbb{E}\left[F'^4_{\Lambda}\right]}}\\&\quad+4\underbrace{\mathbb{E}\left[F^3_{\Lambda}F'_{\Lambda}\right]}_{\le\sqrt{\mathbb{E}\left[F^6_{\Lambda}\right]\mathbb{E}\left[F'^2_{\Lambda}\right]}}+4\underbrace{\mathbb{E}\left[F_{\Lambda}F'^3_{\Lambda}\right]}_{\le\sqrt{\mathbb{E}\left[F^2_{\Lambda}\right]\mathbb{E}\left[F'^6_{\Lambda}\right]}}\\
&\le C<\infty.
\end{align*}
Remark that 
\begin{align}\label{boundsigma}
\E\left[\sup\limits_{\theta\in\Theta}\varepsilon_\theta^2\right]&\le \mathbb{E}\left[(F_{\Lambda}+F'_{\Lambda})^2\right]\le C<\infty.
\end{align}
We have
\begin{align*}
\sqrt{n}\sup\limits_{\theta\in\Theta}|(\mathbb{E}_n-\E )\varepsilon^2_\theta|=\sup\limits_{g\in\mathcal{E}_{\Lambda}^2}|G_n(g)|.
\end{align*}
For every $\sigma^2_C$ with $\sup_{g\in\mathcal{E}_{\Lambda}^2}\E [g^2]\le\sigma^2_C\le \E \big[(F_{\Lambda}+F'_{\Lambda})^4\big]:=G_1<\infty$ and universal constants $K$ and $K_2$ with probability not less than $1-(1/\log (n))$
\begin{align*}
&\quad\sup\limits_{g\in\mathcal{E}_{\Lambda}^2}|G_n(g)|\\&\le 2K\left[\left(S\sigma_C^2\log(AG_1^{1/2}/\sigma_C)\right)^{1/2}+SG_1^{1/2}\log(AG_1^{1/2}/\sigma_C)\right]\\&\quad+K_2(\sigma_C\log(n)^{1/2}+G_1^{1/2}\log(n))\\&=O(\log(n))
\end{align*}
by Lemma 1 in \cite{belloni2014uniform} with $q=2$, $t=\log(n)$, $A=16\sqrt{C_{\Lambda}''}$ and $S=16C_{\Lambda}'$. Therefore it follows with probability $1-o(1)$
$$\sup\limits_{\theta\in\Theta}|(\mathbb{E}_n-\E )\varepsilon^2_\theta|=O\left(\frac{\log(n)}{\sqrt{n}}\right).$$
Analogous it can be shown with probability $1-o(1)$
$$\sup\limits_{\theta\in\Theta}|(\mathbb{E}_n-\E )\Lambda_\theta(Y)^2|=O\left(\frac{\log(n)}{\sqrt{n}}\right).$$
(1) follows with assumption \ref{B2}.\\ \\
Further we have
\begin{align*}
\sup\limits_{|\theta-\theta'|\le 1/n}|\mathbb{E}_n\left[
\varepsilon_\theta-\varepsilon_{\theta'}\right]|=\sup\limits_{|\theta-\theta'|\le 1/n}\frac{1}{\sqrt{n}}|G_n(\varepsilon_\theta-\varepsilon_\theta')|
\end{align*}
Define $\mathcal{E}'_{\Lambda}:=\big\{\varepsilon_\theta-\varepsilon_{\theta'}|\theta,\theta'\in\Theta\big\}$ and $\mathcal{E}_{\Lambda}:=\big\{\varepsilon_\theta=\left(\Lambda_{\theta}(\cdot)-\E\left[\Lambda_{\theta}(\cdot)|X\right]\right)|\theta\in\Theta\big\}$. Using the same argument as above we obtain
\begin{align*}
&\quad\log N(\varepsilon \|2(F_{\Lambda}+F_{\Lambda}')\|_{\tilde{Q},2},\mathcal{E}'_{\Lambda},L_2(\tilde{Q}))\\&\le \log N\left(\frac{\varepsilon}{2} \|2F_{\Lambda}\|_{\tilde{Q},2},\mathcal{F}_{\Lambda},L_2(\tilde{Q})\right)+\log N\left(\frac{\varepsilon}{2} \|2F'_{\Lambda}\|_{\tilde{Q},2},\mathcal{F}'_{\Lambda},L_2(\tilde{Q})\right)\\&\le \sup\limits_{Q}\log N\left(\varepsilon \|F_{\Lambda}\|_{Q,2},\mathcal{F}_{\Lambda},L_2(Q)\right)+\sup\limits_{Q'}\log N\left(\varepsilon \|F'_{\Lambda}\|_{Q',2},\mathcal{F}'_{\Lambda},L_2(Q')\right)\\&\le 2\sup\limits_{Q}\log N\left(\left(\frac{\varepsilon}{4}\right)^2 \|F_{\Lambda}\|_{Q,2},\mathcal{F}_{\Lambda},L_2(Q)\right)\\&\le 4C_{\Lambda}'\log (4\sqrt{C_{\Lambda}''}/\varepsilon).
\end{align*}
Because
$$\mathcal{E}''_{\Lambda}:=\big\{\varepsilon_\theta-\varepsilon_{\theta'}|\theta,\theta'\in\Theta,|\theta-\theta'|\le 1/n\big\}\subset \mathcal{E}'_{\Lambda}$$
we can use Lemma 1 again, since we obtain the same envelope and bound for the entropy as for $\mathcal{E}'_{\Lambda}$.\\
We achieve for every $\sigma^2_n$ with $\sup_{g\in\mathcal{E}''_{\Lambda}}\E [g^2]\le\sigma^2_n\le \E [4(F_{\Lambda}+F'_{\Lambda})^2]:=G_2$ and universal constants $K$ and $K_2$ with probability  at least $1-(1/\log (n))$ 
\begin{align*}
&\quad\sup\limits_{g\in\mathcal{E}''_{\Lambda}}|G_n(g)|\\&\le 2K\left[\left(S\sigma_n^2\log(AG_2^{1/2}/\sigma_n)\right)^{1/2}+n^{-\frac{1}{4}}S2\E[(F_{\Lambda}+F'_{\Lambda})^4]^{1/4}\log(AG_2^{1/2}/\sigma_n)\right]\\&\quad+K_2(\sigma_n\log(n)^{1/2}+n^{-\frac{1}{4}}2\E[(F_{\Lambda}+F'_{\Lambda})^4]^{1/4}\log(n)).
\end{align*}
by Lemma 1 with $q=4$, $t=\log(n)$, $A=4\sqrt{C_{\Lambda}''}$, $S=4C_{\Lambda}'$.\\
We have 
\begin{align*}
&\quad\sup\limits_{|\theta-\theta'|\le \frac{1}{n}}\E [(\varepsilon_\theta-\varepsilon_{\theta'})^2]\\&=\sup\limits_{|\theta-\theta'|\le \frac{1}{n}}\E  \left[\big(\Lambda_\theta(Y)-\E [\Lambda_\theta(Y)|X]-\Lambda_{\theta'}(Y)+\E [\Lambda_{\theta'}(Y)|X]\big)^2 \right] \\
&=\sup\limits_{|\theta-\theta'|\le \frac{1}{n}}\E  \left[\Big(\big(\Lambda_\theta(Y)-\Lambda_{\theta'}(Y)\big)-\big(\E [\Lambda_\theta(Y)|X]-\E [\Lambda_{\theta'}(Y)|X]\big)\Big)^2 \right] \\
&=\sup\limits_{|\theta-\theta'|\le \frac{1}{n}}\Bigg(\E  \left[\big(\Lambda_\theta(Y)-\Lambda_{\theta'}(Y)\big)^2\right]+\E  \bigg[\underbrace{\E \Big[\big(\Lambda_\theta(Y)-\Lambda_{\theta'}(Y)\big)|X\Big]^2}_{\le \E \Big[\big(\Lambda_\theta(Y)-\Lambda_{\theta'}(Y)\big)^2|X\Big] } \bigg] \\
&\quad\quad -2\underbrace{\E  \bigg[\big(\Lambda_\theta(Y)-\Lambda_{\theta'}(Y)\big)\E \Big[\big(\Lambda_\theta(Y)-\Lambda_{\theta'}(Y)\big)|X\Big]\bigg]}_{\ge 0}\Bigg)\\
&\le \sup\limits_{|\theta-\theta'|\le \frac{1}{n}}2\E  \left[\big(\Lambda_\theta(Y)-\Lambda_{\theta'}(Y)\big)^2\right]\\
&= \sup\limits_{|\theta-\theta'|\le \frac{1}{n}}2\E \Big[(\theta-\theta')^2\big(\dot\Lambda_{\bar{\theta}}(Y)\big)^2\Big]\\
&\le \frac{2}{n^2}\underbrace{\sup\limits_{\theta\in\Theta}\E \Big[\big(\dot\Lambda_\theta(Y)\big)^2\Big]}_{\le C}=O(n^{-2}).
\end{align*}\ \\
Therefore we can choose $\sigma_n^2=O(n^{-2})$ and obtain obtain with probability $1-o(1)$
\begin{align*}
n^{1/2}\sup\limits_{|\theta-\theta'|\le 1/n}|\mathbb{E}_n\left[
\varepsilon_\theta-\varepsilon_{\theta'}\right]|&=\sup\limits_{|\theta-\theta'|\le 1/n}|G_n(\varepsilon_\theta-\varepsilon_\theta')|\\
&=\sup\limits_{g\in\mathcal{E}''_\Lambda}|G_n(g)|\\
&=O\left(\frac{\log(n)}{n^{1/4}}\right)=O\left(\delta_n\right).
\end{align*}
For (3) we can use the same arguments as above and we remark
\begin{align*}
\sup\limits_{|\theta-\theta'|\le 1/n}\mathbb{E}_n\left[(\varepsilon_\theta-\varepsilon_{\theta'})^2\right]&\le \sup\limits_{|\theta-\theta'|\le 1/n}\mathbb{E}\left[(\varepsilon_\theta-\varepsilon_{\theta'})^2\right]\\&\quad+\left|\sup\limits_{|\theta-\theta'|\le 1/n}\left(\mathbb{E}_n\left[(
\varepsilon_\theta-\varepsilon_{\theta'})^2\right]-\mathbb{E}\left[(
\varepsilon_\theta-\varepsilon_{\theta'})^2\right]\right)\right|\\
&\le \sup\limits_{g\in{\mathcal{E}_{\Lambda}^2}'}\frac{1}{\sqrt{n}}G_n(g)+O(n^{-2})
\end{align*}
with ${\mathcal{E}_{\Lambda}^2}':=\big\{(\varepsilon_\theta-\varepsilon_{\theta'})^2|\theta,\theta'\in\Theta\big\}.$ The entropy of this class is bounded by
\begin{align*}
&\quad\log N(\varepsilon \|4(F_{\Lambda}+F_{\Lambda}')^2\|_{\tilde{Q},2},{\mathcal{E}_{\Lambda}^2}',L_2(\tilde{Q}))\\&\le 2\log N\left(\frac{\varepsilon}{2} \|4(F_{\Lambda}+F_{\Lambda}')\|_{\tilde{Q},2},\mathcal{E}_{\Lambda}',L_2(\tilde{Q})\right)\\&\le 2\log N\left(\frac{\varepsilon}{4} \|4F_{\Lambda}\|_{\tilde{Q},2},\mathcal{F}_{\Lambda},L_2(\tilde{Q})\right)+2\log N\left(\frac{\varepsilon}{4} \|4F'_{\Lambda}\|_{\tilde{Q},2},\mathcal{F}'_{\Lambda},L_2(\tilde{Q})\right)\\&\le 2\sup\limits_{Q}\log N\left(\varepsilon \|F_{\Lambda}\|_{Q,2},\mathcal{F}_{\Lambda},L_2(Q)\right)+2\sup\limits_{Q'}\log N\left(\varepsilon \|F'_{\Lambda}\|_{Q',2},\mathcal{F}'_{\Lambda},L_2(Q')\right)\\&\le 4\sup\limits_{Q}\log N\left(\left(\frac{\varepsilon}{4}\right)^2 \|F_{\Lambda}\|_{Q,2},\mathcal{F}_{\Lambda},L_2(Q)\right)\\&\le 8C_{\Lambda}'\log \left(4\sqrt{C_{\Lambda}''}/\varepsilon\right).
\end{align*}
For every $\sigma^2_C$ with $\sup_{g\in{\mathcal{E}_{\Lambda}^2}'}\E [g^2]\le\sigma^2_C\le \E \big[16(F_{\Lambda}+F'_{\Lambda})^4\big]:=G_3<\infty$ and universal constants $K$ and $K_2$ with probability not less than $1-(1/\log (n))$, it holds
\begin{align*}
&\quad\sup\limits_{g\in{\mathcal{E}_{\Lambda}^2}'}|G_n(g)|\\&\le 2K\left[\left(S\sigma_C^2\log(AG_3^{1/2}/\sigma_C)\right)^{1/2}+SG_3^{1/2}\log(AG_3^{1/2}/\sigma_C)\right]\\&\quad+K_2(\sigma_C\log(n)^{1/2}+G_3^{1/2}\log(n))\\&=O(\log(n))
\end{align*}
by Lemma 1 in \cite{belloni2014uniform} with $q=2$, $t=\log(n)$, $A=4\sqrt{C_{\Lambda}''}$ and $S=8C_{\Lambda}'$.\\
We conclude
$$\sup\limits_{|\theta-\theta'|\le 1/n}\mathbb{E}_n\left[(
\varepsilon_\theta-\varepsilon_{\theta'})^2\right]=O\left(\frac{\log n}{\sqrt{n}}\right)$$
and therefore
\begin{align*}
\log(p\vee n)^{1/2}\sup\limits_{|\theta-\theta'|\le 1/n}\mathbb{E}_n\left[(
\varepsilon_\theta-\varepsilon_{\theta'})^2\right]^{1/2}&=O(\delta_n)
\end{align*}
since $n^{-1/4}\log(p\vee n)\le\delta_n$.
\end{proof}
\section{Inference on a target parameter in Z-problems when the high-dimensional nuisance functions depends on the target parameter}\label{appendixC}
In this section we consider a general Z-problem, where target parameter $\theta_0$ obeys the moment condition
$$\E \Big[\psi\big((Y,X),\theta_0,h_0(\theta_0,X)\big)\Big]=0.$$
In this setting the unkown nuisance function $h_0(\theta,X)=(h_{0,1}(\theta,X),\dots,h_{0,m}(\theta,X))\in\mathcal{H}$ may depend on $\theta$. 
The central theorem is a statement about the asymptotic distribution of an estimate which solves
\begin{align}\label{Zestimate}
\left|\mathbb{E}_n\Big[\psi\big((Y,X),\hat{\theta},\hat{h}_0(\hat{\theta},X)\big)\Big]\right|=\inf_{\theta\in\Theta}\left|\mathbb{E}_n\Big[\psi\big((Y,X),\theta,\hat{h}_0(\theta,X)\big)\Big]\right|+\epsilon_n.
\end{align}
We need a more general form of the conditions in section \ref{main section}.\\
Assumptions \textbf{C1}-\textbf{C7}.\\ 
The following assumptions hold uniformly in $n\ge n_0,P\in\mathcal{P}_n$:
\begin{enumerate}[label=\textbf{C\arabic*},ref=C\arabic*]
\item\label{C1}
The true parameter $\theta_0$ obeys the moment condition
$$\E \Big[\psi\big((Y,X),\theta_0,h_0\big)\Big]=0.$$
\item \label{C3}
The map $(\theta,h)\mapsto\E[\psi((X,Y),\theta,h)]$ is twice continuously Gateaux-differentiable on $\Theta\times\mathcal{H}$.
\item \label{C4}
Let $\tilde{\mathcal{H}}=\{\tilde{h}:\Theta\times\mathcal{X}\mapsto\mathbb{R}^m\}\subseteq\mathcal{H}$ be a suitable set of functions. For every $\theta\in\Theta$ we have a nuissance function estimator $\hat{h}(\theta)$ and a set of functions $\tilde{\mathcal{H}}(\theta)=\{\tilde{h}:\mathcal{X}\mapsto\mathbb{R}^m:\tilde{h}(x)=\tilde{h}(\theta,x)\in\tilde{\mathcal{H}}\}$ with $P(\hat{h}(\theta)\in\tilde{\mathcal{H}}(\theta))=1-o(1)$, where $\tilde{\mathcal{H}}(\theta)$ contains $h_0(\theta,\cdot)$ and is constrained by conditions given below.
\item \label{C2}
For all $\tilde{h}\in \tilde{\mathcal{H}}$ the score $\psi$ obeys the Neyman orthogonality property
$$D_0[\tilde{h}-h_0]=0.$$
\item \label{C5}
For all $\theta\in\Theta$ the class of functions
$$\Psi(\theta)=\Big\{(y,x)\mapsto \psi\big((y,x),\theta,\tilde{h}(\theta,x)\big),\tilde{h}\in\tilde{\mathcal{H}}(\theta)\Big\}$$
has a measurable envelope $\bar\psi\ge \sup_{\psi\in \Psi(\theta)}|\psi|$ independent from $\theta$, such that for some $q\ge 4$
$$\E \Big[(\bar\psi(Y,X))^q\Big]\le C.$$
The class $\Psi(\theta)$ is pointwise measurable and uniformly for all $\theta\in\Theta$
$$\sup_{Q}\log N(\varepsilon||\bar\psi||_{Q,2},\Psi(\theta),L_2(Q))\le C_1 s\log\left(\frac{C_2(p\vee n)}{\varepsilon}\right)$$
with $C_1$ and $C_2$ beeing independent from $\theta$.
\item \label{C6}
\begin{itemize}
\item[(i)] For a sequence $\rho_n$ with $$n^{-1/2}\left(s^{\frac{1}{2}}\log(p\vee n)^{\frac{1}{2}}+n^{-\frac{1}{2}+\frac{1}{q}}s\log(p\vee n) \right)=O(\rho_n)$$ we have
$$\sup\limits_{\theta\in\Theta,\tilde{h}\in\tilde{\mathcal{H}}(\theta)}|\E[\psi((Y,X)),\theta,h_0(\theta,X)]-\E[\psi((Y,X)),\theta,\tilde{h}(\theta,X)]|\le C\rho_n.$$
\item[(ii)] We define
\begin{align*}
\sup\E\Big[\Big(\psi\big((Y,X),\theta,\tilde{h}(\theta,X)\big)-\psi\big((Y,X),\theta_0,h_0(\theta_0,X)\big)\Big)^2\Big]^{1/2}=: r_n,
\end{align*}
where the supremum is taken over all $\theta$ with $|\theta-\theta_0|\le C\rho_n$ and $\tilde{h}\in\tilde{\mathcal{H}}$, meaning 
$$\sup\equiv\sup\limits_{\theta:|\theta-\theta_0|\le C\rho_n,\tilde{h}\in\tilde{\mathcal{H}}(\theta)},$$
and it holds $r_n s^{\frac{1}{2}}\log\Big(\frac{(p\vee n) }{r_n}\Big)^{\frac{1}{2}}+n^{-\frac{1}{2}+\frac{1}{q}}s\log\Big(\frac{(p\vee n) }{r_n}\Big)=o(1)$ with $q$ from assumption C$5$.
\item[(iii)] It holds
\begin{align*}\sup\left|\partial^2_r\bigg\{\E\Big[\psi\big((Y,X),\theta_0+r(\theta-\theta_0),h_0+r(\tilde{h}-h_0)\big)\Big]\bigg\}\right|=o(n^{-1/2}),
\end{align*}
where
$$\sup\equiv\sup\limits_{r\in(0,1),\theta:|\theta-\theta_0|\le C\rho_n,\tilde{h}\in\tilde{\mathcal{H}}(\theta)}.$$
\end{itemize}
\item \label{C7}
For $h \in \tilde{\mathcal{H}}$ the function
$$\theta\mapsto \E \Big[\psi\big((Y,X),\theta,h(\theta,X)\big)\Big]$$
is differentiable in a neighborhood of $\theta_0$ and
for all $\theta\in\Theta$, the identification relation
$$2|\E[\psi((Y,X)),\theta,h_0(\theta,X)]|\ge |\Gamma(\theta-\theta_0)|\wedge c_0$$
is satisfied with
$$\Gamma:=\partial_\theta\E \Big[\psi\big((Y,X),\theta_0,h_0(\theta_0,X)\big)\Big]>c_1.$$
\end{enumerate}
\noindent
Since the nuisance functions depend on the target parameter the conditions ensure that they can be estimated uniformly over all $\theta$ with a sufficiently fast rate. 
\begin{theorem}{} \label{generalmain}
Under the assumptions \ref{C1}-\ref{C7} an estimator $\hat{\theta}$ of the form in \ref{Zestimate} obeys
$$n^{\frac{1}{2}}(\hat{\theta}-\theta_0)\xrightarrow{\mathcal{D}}\mathcal{N}(0,\Sigma),$$
where $$\Sigma:=\E\Big[\Gamma^{-2}\psi^2\big((Y,X),\theta_0,h_0(\theta_0,X)\big)\Big]$$ with $\Gamma=\partial_\theta\E \Big[\psi\big((Y,X),\theta_0,h_0(\theta_0,X)\big)\Big]$.
\end{theorem}

\begin{remark}\label{entropycomment}   \ \\
This setting and the theorem is almost identical to the assumption 3.4 and theorem 3.3 in Chernozhukov et al. (2017) \cite{chernozhukov2017double}. Their theorem holds for dependent nuisance functions, but the entropy condition may be hard to verify in some settings:\\
Suppose the unknown nuisance function $h_0$ is a linear function of $X$, where the coefficients $\beta_0(\theta)$ ($\|\beta_0(\theta)\|_0\le s$ for all $\theta$) are dependent on the target parameter. If $h_0(\theta,X)=X\beta_0(\theta)$ is estimated with the lasso estimator the uniform covering entropy of 
$$\mathcal{F}_h:=\Big\{\psi\big(\cdot,\theta,h(\theta,\cdot)\big),\theta\in\Theta\Big\}$$
may not fulfill the desired condition. This is because the uniform covering entropy of the class
$$\mathcal{H}:=\Big\{h(\theta,\cdot):\mathcal{X}\to \mathbb{R}|h(\theta,X)=\beta(\theta) X,\|\beta(\theta)\|_0\le s,\theta\in\Theta\Big\}$$
can not be bounded by standard arguments as the union over sets with a bounded VC-index (see for example Belloni et al. (2014) \cite{belloni2014uniform}), since the indices which are not zero may vary for each $\theta$.\\
In their example, the estimation of the average treatment effect, this problem does not occur, because the estimated nuisance functions do not depend on the target parameter.\\
To bypass this we rely on a slightly different set of entropy conditions, which enables us to restrict the entropy of the classes for an arbitrary $\theta$, but uniformly over all $\theta\in\Theta$. 
\end{remark}
\begin{proof}
We mimic the proof of Theorem 2 from Belloni, Chernozhukov and Kato (2014) \cite{belloni2014uniform}.\\
We prove the theorem under an abitraty sequence $P=P_n\in\mathcal{P}_n$. Therefore the dependence of $P$ on $n$ can be suppressed.\\ \\
\textit{Step 1.}\\
Let $\tilde{\theta}$ be an abitrary estimator fullfilling $|\tilde{\theta}-\theta_0|\le C\rho_n$ with probability $1-o(1)$. We aim to prove that with probability $1-o(1)$
\begin{align*}
&\quad\E_n\Big[\psi\big((Y,X),\tilde{\theta},\hat{h}(\tilde{\theta},X)\big)\Big]\\&=\E_n\Big[\psi\big((Y,X),\theta_0,h_0(\theta_0,X)\big)\Big]\\
&\quad+\underbrace{\partial_\theta\E \Big[\psi\big((Y,X),\theta_0,h_0(\theta_0,X)\big)\Big]}_{:=\Gamma}(\tilde{\theta}-\theta_0)+o(n^{-\frac{1}{2}}). 
\end{align*}
By assumption \ref{C1} we can expand the term
\begin{align*}
&\quad\E_n\Big[\psi\big((Y,X),\tilde{\theta},\hat{h}(\tilde{\theta},X)\big)\Big]\\&=\E_n\Big[\psi\big((Y,X),\tilde{\theta},\hat{h}(\tilde{\theta},X)\big)\Big]+\underbrace{\E \Big[\psi\big((Y,X),\theta_0,h_0(\theta_0,X)\big)\Big]}_{=0}\\
&\quad+\E_n \Big[\psi\big((Y,X),\theta_0,h_0(\theta_0,X)\big)\Big]-\E_n \Big[\psi\big((Y,X),\theta_0,h_0(\theta_0,X)\big)\Big]\\
&\quad+\E\Big[\psi\big((Y,X),\tilde{\theta},\hat{h}(\tilde{\theta},X)\big)\Big]-\E\Big[\psi\big((Y,X),\tilde{\theta},\hat{h}(\tilde{\theta},X)\big)\Big]\\
&=\underbrace{\E_n \Big[\psi\big((Y,X),\theta_0,h_0(\theta_0,X)\big)\Big]}_{:=I}+\underbrace{\E\Big[\psi\big((Y,X),\tilde{\theta},\hat{h}(\tilde{\theta},X)\big)\Big]}_{:=II}\\
&\quad +\underbrace{\E_n\Big[\psi\big((Y,X),\tilde{\theta},\hat{h}(\tilde{\theta},X)\big)\Big]-\E\Big[\psi\big((Y,X),\tilde{\theta},\hat{h}(\tilde{\theta},X)\big)\Big]}_{:=III}\\
&\quad-\bigg(\underbrace{\E_n \Big[\psi\big((Y,X),\theta_0,h_0(\theta_0,X)\big)\Big]-\E \Big[\psi\big((Y,X),\theta_0,h_0(\theta_0,X)\big)\Big]}_{:=IV}\bigg)\\
&=I+II+III-IV
\end{align*}
Considering the last two terms we have with probability $1-o(1)$
\begin{align*}
&\quad n^{\frac{1}{2}}\Big(III-IV\Big)\\&=\frac{1}{\sqrt{n}}\sum\limits_{i=1}^n\bigg(\psi\big((Y,X),\tilde{\theta},\hat{h}(\tilde{\theta},X)\big)-\psi\big((Y,X),\theta_0,h_0(\theta_0,X)\big)\\
&\quad -\Big(\E\Big[\psi\big((Y,X),\tilde{\theta},\hat{h}(\tilde{\theta},X)\big)\Big]-\E \Big[\psi\big((Y,X),\theta_0,h_0(\theta_0,X)\big)\Big]\Big)\bigg)\\
&\le \sup\limits_{\theta:|\theta-\theta_0|\le C\rho_n}\bigg|\bigg[\frac{1}{\sqrt{n}}\sum\limits_{i=1}^n\bigg(\psi\big((Y,X),\theta,\hat{h}(\theta,X)\big)-\psi\big((Y,X),\theta_0,h_0(\theta_0,X)\big)\\
&\quad -\Big(\E\Big[\psi\big((Y,X),\theta,\hat{h}(\theta,X)\big)\Big]-\E \Big[\psi\big((Y,X),\theta_0,h_0(\theta_0,X)\big)\Big]\Big)\bigg)\bigg]\bigg|\\
&\le \sup\limits_{\theta:|\theta-\theta_0|\le C\rho_n}\left(\sup\limits_{f\in\Psi'(\theta)}|G_n(f)|\right)
\end{align*}
with
$$\Psi'(\theta)=\Big\{(y,x)\mapsto \psi\big((y,x),\theta,\tilde{h}(\theta,x)\big)-\psi\big((y,x),\theta_0,h_0(\theta_0,x)\big),\tilde{h}\in\tilde{\mathcal{H}}(\theta)\Big\}$$
and envelope $2\bar{\psi}$. Here we used assumption \ref{C5} and that with probability $1-o(1)$ we have $\hat{h}(\theta,X),h_0(\theta,X)\in\tilde{\mathcal{H}}(\theta)$ for all $\theta\in\Theta$ by assumption \ref{C4}.
Recall that
$$\sup_{Q}\log N(\varepsilon||2\bar\psi||_{Q,2},\Psi'(\theta),L_2(Q))\le C_1 s\log\left(\frac{C_2(p\vee n)}{\varepsilon}\right)$$
for constants $C_1$ and $C_2$ beeing independent from $\theta$.
We want to apply Lemma 1 from Belloni, Chernozhukov and Kato (2014) \cite{belloni2014uniform}.
By assumption \ref{C6} we have
\begin{align*}
&\quad\sup\limits_{\theta:|\theta-\theta_0|\le C\rho_n,f\in\Psi'(\theta)}\E \Big[f^2\big((Y,X)\big)\Big]\\
&=\sup\limits_{\theta:|\theta-\theta_0|\le C\rho_n,\tilde{h}\in\tilde{\mathcal{H}}(\theta)}\E\Big[\Big(\psi\big((Y,X),\theta,\tilde{h}(\theta,X)\big)-\psi\big((Y,X),\theta_0,h_0(\theta_0,X)\big)\Big)^2\Big]\\
&=: r_n^2
\end{align*}
with $r_n s^{\frac{1}{2}}\log\Big(\frac{(p\vee n) }{r_n}\Big)^{\frac{1}{2}}+n^{-\frac{1}{2}+\frac{1}{q}}s\log\Big(\frac{(p\vee n) }{r_n}\Big)=o(1)$. Choosing $\sigma_n^2=r_n^2$ and $\max_{q\in\{2,4\}}\E [(\bar\psi(Y,X))^q]\le C$ the first inequality of Lemma 1 in \cite{belloni2014uniform} implies
\begin{align*}
&\quad\E \left[\sup\limits_{f\in\Psi'(\theta)}|G_n(f)|\right]\\&\le K\bigg[\bigg(C_1 s  \sigma_n^2\log\Big(\frac{C_2(p\vee n) C^{\frac{1}{2}}}{\sigma_n}\Big)\bigg)^{\frac{1}{2}} + n^{-\frac{1}{2}+\frac{1}{q}}C_1s C^{\frac{1}{q}}\log\Big(\frac{C_2(p\vee n) C^{\frac{1}{2}}}{\sigma_n}\Big)\bigg]\\
&\le K'\bigg( \sigma_n\bigg( s  \log\Big(\frac{(p\vee n) }{\sigma_n}\Big)\bigg)^{\frac{1}{2}}+n^{-\frac{1}{2}+\frac{1}{q}}s\log\Big(\frac{(p\vee n) }{\sigma_n}\Big)\bigg).
\end{align*}
Applying the second part of Lemma 1 with $t=\log(n)$ we obtain
\begin{align*}
n^{\frac{1}{2}}|III-IV|&\le \sup\limits_{\theta:|\theta-\theta_0|\le C\rho_n}\Big(\sup\limits_{f\in\Psi'(\theta)}|G_n(f)|\Big)\\& \le \sup\limits_{\theta:|\theta-\theta_0|\le C\rho_n}\bigg(2\E \Big[\sup\limits_{f\in\Psi'(\theta)}|G_n(f)|\Big]\\&\quad+K_q\bigg(\sigma_n\log(n)^{\frac{1}{2}}+n^{-\frac{1}{2}+\frac{1}{q}}C^{\frac{1}{q}}\log(n)\bigg)\bigg)\\
&\le K_q'\bigg( \sigma_n\bigg( s  \log\Big(\frac{(p\vee n) }{\sigma_n}\Big)\bigg)^{\frac{1}{2}}+n^{-\frac{1}{2}+\frac{1}{q}}s\log\Big(\frac{(p\vee n) }{\sigma_n}\Big)\bigg)\\&=o(1).
\end{align*}
Now we expand the term II.  Let $\tilde{h}\in\tilde{\mathcal{H}}$ and $\tilde{\theta}\in\Theta$.\\
By Taylor expansion of the function $r\mapsto\E\Big[\psi\big((Y,X),\theta_0+r(\tilde{\theta}-\theta_0),h_0+r(\tilde{h}-h_0)\big)\Big]$ we have by assumption \ref{C3}
\begin{align*}
&\quad\E\Big[\psi\big((Y,X),\tilde{\theta},\tilde{h}\big)\Big]\\&=\E\Big[\psi\big((Y,X),\theta_0,h_0\big)\Big]\\
&\quad+\partial_r\bigg\{\E\Big[\psi\big((Y,X),\theta_0+r(\tilde{\theta}-\theta_0),h_0+r(\tilde{h}-h_0)\big)\Big]\bigg\}\bigg|_{r=0}\\
&\quad+\frac{1}{2}\partial^2_r\bigg\{\E\Big[\psi\big((Y,X),\theta_0+r(\tilde{\theta}-\theta_0),h_0+r(\tilde{h}-h_0)\big)\Big]\bigg\}\bigg|_{r=\bar{r}},
\end{align*}
for some $\bar{r}\in(0,1)$. Due to the orthogonality condition in \ref{C2} we have
\begin{align*}
\quad& \partial_r\bigg\{\E\Big[\psi\big((Y,X),\theta_0+r(\tilde{\theta}-\theta_0),h_0+r(\tilde{h}-h_0)\big)\Big]\bigg\}\bigg|_{r=0}\\
=& \partial_r\bigg\{\E\Big[\psi\big((Y,X),\theta_0+r(\tilde{\theta}-\theta_0),h_0+r(\tilde{h}-h_0)\big)\Big]\bigg\}\bigg|_{r=0}-D_0[\tilde{h}-h_0]\\
=&\partial_r\bigg\{\E\Big[\psi\big((Y,X),\theta_0+r(\tilde{\theta}-\theta_0),h_0+r(\tilde{h}-h_0)\big)\Big]\\&-\E\Big[\psi\big((Y,X),\theta_0,h_0+r(\tilde{h}-h_0)\big)\Big]\bigg\}\bigg|_{r=0}\\
=&\partial_r\bigg\{r(\tilde{\theta}-\theta_0)\partial_\theta\E\Big[\psi\big((Y,X),\theta,h_0+r(\tilde{h}-h_0)\big)\Big]\Big|_{\theta_\in[\theta_0,\theta_0+r(\tilde{\theta}-\theta_0)]}\bigg\}\bigg|_{r=0}\\
=&(\tilde{\theta}-\theta_0)\partial_\theta\E\Big[\psi\big((Y,X),\theta_0,h_0\big)\Big].
\end{align*}
By assumption \ref{C6} we have
\begin{align*}
\left|\partial^2_r\bigg\{\E\Big[\psi\big((Y,X),\theta_0+r(\tilde{\theta}-\theta_0),h_0+r(\tilde{h}-h_0)\big)\Big]\bigg\}\bigg|_{r=\bar{r}}\right|=o(n^{-1/2})
\end{align*}
and therefore
\begin{align*}
\E\Big[\psi\big((Y,X),\tilde{\theta},\tilde{h}\big)\Big]=\Gamma(\tilde{\theta}-\theta_0)+o(n^{-1/2}).
\end{align*}
In total, we obtain with probability $1-o(1)$
\begin{align*}
\E_n\Big[\psi\big((Y,X),\tilde{\theta},\hat{h}(\tilde{\theta},X)\big)\Big]&=\E_n\Big[\psi\big((Y,X),\theta_0,h_0(\theta_0,X)\big)\Big]\\
&\quad+\Gamma(\tilde{\theta}-\theta_0)+o(n^{-\frac{1}{2}}). 
\end{align*}
\newpage\noindent
\textit{Step 2.}\\
We want to prove that with probability $1-o(1)$
$$\inf\limits_{\theta\in\Theta}\Big|\E_n\Big[\psi\big((Y,X),\theta,\hat{h}(\theta,X)\big)\Big]\Big|=o(n^{-\frac{1}{2}}).$$
Define
$$\theta^*:=\theta_0-\Gamma^{-1}\E_n\Big[\psi\big((Y,X),\theta_0,h_0(\theta_0,X)\big)\Big].$$
Directly follows with central limit theorem
$$|\theta^*-\theta_0|= |\Gamma^{-1}|\Big|\underbrace{\E_n\Big[\psi\big((Y,X),\theta_0,h_0(\theta_0,X)\big)\Big]}_{=O\left(n^{-\frac{1}{2}}\right)}\Big|\le C \rho_n.$$
Using Step 1 we obtain with probability $1-o(1)$
$$\inf\limits_{\theta\in\Theta}\Big|\E_n\Big[\psi\big((Y,X),\theta,\hat{h}(\theta,X)\big)\Big]\Big|\le \Big|\E_n\Big[\psi\big((Y,X),\theta^*,\hat{h}(\theta^*,X)\big)\Big]\Big|= o(n^{-\frac{1}{2}})$$
by inserting the definition of $\theta^*$.\\ \\
\textit{Step 3.}\\
We aim to show that the estimated $\hat{\theta}$ converges towards $\theta_0$, meaning with probability $1-o(1)$
$$|\hat{\theta}-\theta_0|\le C\rho_n.$$
By definition of $\hat{\theta}$ and Step 2 we have
$$\Big|\E_n\Big[\psi\big((Y,X),\hat{\theta},\hat{h}(\hat{\theta},X)\big)\Big]\Big|= o(n^{-\frac{1}{2}}).$$
Since $\hat{h}(\theta)\in\tilde{\mathcal{H}}(\theta)$ with probability $1-o(1)$ for all $\theta\in\Theta$ we have 
\begin{align*}
&\quad\sup\limits_{\theta\in\Theta}\left|\E_n\Big[\psi\big((Y,X),\theta,\hat{h}(\theta,X)\big)\Big]-\E\Big[\psi\big((Y,X),\theta,\hat{h}(\theta,X)\big)\Big]\right|\\
&\le  \sup\limits_{\theta\in\Theta}\left(n^{-\frac{1}{2}}\sup\limits_{g\in\Psi(\theta)}|G_n(g)|\right)\\&=O\left(n^{-1/2}\left(s^{\frac{1}{2}}\log(p\vee n)^{\frac{1}{2}}+n^{-\frac{1}{2}+\frac{1}{q}}s\log(p\vee n)\right)\right),
\end{align*}
where we used Lemma 1 in \cite{belloni2014uniform} and $\E \Big[(\bar\psi(Y,X))^2\Big]\le C$ as in Step 1. Combining this with the triangle inequality we obtain
\begin{align*}
&\quad\big|\E\Big[\psi\big((Y,X),\hat{\theta},h_0(\hat{\theta},X)\big)\Big]\big|\\
&\le\sup\limits_{\theta\in\Theta,\tilde{h}\in\tilde{\mathcal{H}}(\theta)}|\E[\psi((Y,X)),\theta,h_0(\theta,X)]-\E[\psi((Y,X)),\theta,\tilde{h}(\theta,X)]|\\
&+\sup\limits_{\theta\in\Theta,\tilde{h}(\theta)\in\tilde{\mathcal{H}}(\theta)}\left|\E_n\Big[\psi\big((Y,X),\theta,\tilde{h}(\theta,X)\big)\Big]-\E\Big[\psi\big((Y,X),\theta,\tilde{h}(\theta,X)\big)\Big]\right|\\
&+
\Big|\E_n\Big[\psi\big((Y,X),\hat{\theta},\hat{h}(\hat{\theta},X)\big)\Big]\Big|\le C\rho_n
\end{align*}
by \ref{C6}. Hence, it follows by assumption \ref{C7} with probability $1-o(1)$,
\begin{align*}
|\Gamma(\hat{\theta}-\theta_0)|\wedge c_0 \le 2\big|\E[\psi((Y,X)),\hat{\theta},h_0(\hat{\theta},X)]\big|\le C\rho_n
\end{align*}
and dividing by $\Gamma>c_1$ gives the claim of this step.\\ \\ 
\textit{Step 4.}\\
Because of Step 3 we are able to use Step 1 for the estimated parameter and obtain with probability $1-o(1)$
\begin{align*}
\E_n\Big[\psi\big((Y,X),\hat{\theta},\hat{h}(\hat{\theta},X)\big)\Big]&=\E_n\Big[\psi\big((Y,X),\theta_0,h_0(\theta_0,X)\big)\Big]\\
&\quad+\Gamma(\hat{\theta}-\theta)+o\left(n^{-\frac{1}{2}}\right).
\end{align*}
By Step 2 we have
\begin{align*}
&\quad\Gamma(\hat{\theta}-\theta)\\
&=\underbrace{\E_n\Big[\psi\big((Y,X),\hat{\theta},\hat{h}(\hat{\theta},X)\big)\Big]}_{=o\left(n^{-\frac{1}{2}}\right)}-\E_n\Big[\psi\big((Y,X),\theta_0,h_0(\theta_0,X)\big)\Big]+o\left(n^{-\frac{1}{2}}\right)\\
&=-\bigg(\E_n\Big[\psi\big((Y,X),\theta_0,h_0(\theta_0,X)\big)\Big]-\underbrace{\E\Big[\psi\big((Y,X),\theta_0,h_0(\theta_0,X)\big)\Big]}_{=0}\bigg)\\
&\quad+o\left(n^{-\frac{1}{2}}\right)
\end{align*}
Using the central limit theorem we get with probability $1-o(1)$
\begin{align*}
&\quad n^{\frac{1}{2}}(\hat{\theta}-\theta) \\
&=\underbrace{-\Gamma^{-1}n^{\frac{1}{2}}\bigg(\E_n\Big[\psi\big((Y,X),\theta_0,h_0(\theta_0,X)\big)\Big]-\E\Big[\psi\big((Y,X),\theta_0,h_0(\theta_0,X)\big)\Big]\bigg)}_{\xrightarrow{\mathcal{D}}\mathcal{N}(0,\Sigma)}\\
&\quad\quad+o(1)
\end{align*}
with $$\Sigma:=\Var\Big(\Gamma^{-1}\psi\big((Y,X),\theta_0,h_0(\theta_0,X)\big)\Big)=\E\Big[\Gamma^{-2}\psi^2\big((Y,X),\theta_0,h_0(\theta_0,X)\big)\Big].$$
\end{proof}

\newpage

\footnotesize
\pagebreak
\bibliographystyle{imsart-number}
\bibliography{Literatur_NR}

\end{document}